\documentclass[journal]{IEEEtran}
\usepackage{amsmath,amsfonts,bm,amsthm}
\usepackage{algorithmic}
\usepackage{algorithm}
\usepackage{array}
\usepackage{subfig}
\usepackage{textcomp, xcolor}
\usepackage{stfloats}
\usepackage{url}
\usepackage{verbatim}
\usepackage{graphicx}
\usepackage{cite}
\usepackage{accents,etoolbox}

\newcommand\munderbar[1]{
  \underaccent{\bar}{#1}}
\DeclareMathOperator*{\argmin}{arg\,min}
\DeclareMathOperator*{\argmax}{arg\,max}

\newtheorem{theorem}{Theorem}
\makeatletter
\newtheoremstyle{nonparen}% name
  {3pt}{3pt}% Space above/below
  {\itshape}{}% Body font / indent
  {\bfseries}{.}{.5em}% Head font / punct / space
  {\thmname{#1}\ \thmnumber{#2}{\normalfont\thmnote{~#3}}}% note 非粗体且用方括号
\makeatother
\theoremstyle{nonparen}
\newtheorem{lemma}{Lemma}

\begin{document}
\bstctlcite{IEEEexample:BSTcontrol}
\title{Zak-OTFS ISAC with Bistatic Sensing via Semi-Blind Atomic Norm Denoising Scheme}

\author{Kecheng Zhang, Weijie Yuan,~\IEEEmembership{Senior Member,~IEEE}, Maria Sabrina Greco,~\IEEEmembership{Fellow,~IEEE}
        \thanks{Kecheng Zhang and Weijie Yuan are with Guangdong Provincial Key Laboratory of Fully Actuated System Control Theory and Technology, School of Automation and Intelligent Manufacturing, Southern University of Science and Technology, 518055 Shenzhen, China. (email: zhangkc2022@mail.sustech.edu.cn; yuanwj@sustech.edu.cn).}
        \thanks{Maria Sabrina Greco is with the Department of Information Engineering, University of Pisa, 56122 Pisa, Italy (e-mail: maria.greco@unipi.it)}
        }

\maketitle

\begin{abstract}
        Zak-transform-based orthogonal time frequency space (Zak-OTFS) modulation provides a delay-Doppler (DD) domain framework for integrated sensing and communication (ISAC) in high-mobility scenarios. Since the DD-domain channel response is directly related to the target parameters and also determines the communication link, accurate channel estimation is a key task for ISAC. However, it is challenging due to the fractional delay and Doppler shifts, which spread the channel response beyond the on-grid DD bins and lead to strong coupling between channel estimation and data detection. To address this issue, this paper proposes a semi-blind atomic norm denoising scheme for Zak-OTFS ISAC with bistatic sensing. We first derive the discrete-time input-output (I/O) relationship of Zak-OTFS with rectangular windowing. Based on this I/O relation, the joint channel parameter estimation and data detection problem is formulated as an atomic norm denoising problem, where a negative square penalty is introduced to handle the non-convex discrete constellation constraints. An accelerated iterative algorithm is then developed by combining majorization-minimization, accelerated projected gradient, and inexact accelerated proximal gradient methods. We also establish the convergence of the proposed algorithm. Simulation results show that the proposed scheme achieves super-resolution sensing accuracy and communication performance close to the perfect-CSI lower bound.
\end{abstract}


\section{Introduction}
The growing demand for reliable wireless communication in high-mobility scenarios has motivated the research for modulation schemes that operate effectively in doubly dispersive channels \cite{lu2024integrated}. While Orthogonal Frequency Division Multiplexing (OFDM) forms the foundation of 4G and 5G networks \cite{10596930, 10969844}, its subcarrier orthogonality is easily disrupted in high-mobility multipath channels, leading to severe inter-carrier interference and performance degradation. Orthogonal Time Frequency Space (OTFS) modulation \cite{OTFS0} has emerged as a promising solution to these challenges. By mapping symbols to the delay-Doppler (DD) domain, OTFS leverages the domain's inherent sparsity and stability to achieve full diversity \cite{Fair_Compare_OTFS_OFDM_Giuseppe}, significantly outperforming OFDM in high-mobility scenarios \cite{OTFS0, Fair_Compare_OTFS_OFDM_Giuseppe}. Moreover, the delays and Doppler shifts of the DD-domain channel directly characterize the sensing targets, while the complete channel response is required for reliable data detection. Therefore, both the sensing and communication performance of OTFS-ISAC depend on accurate DD-domain channel estimation. In this paper, we consider an OTFS-ISAC system with bistatic sensing, where the receiver knows the pilot and frame structure but not the randomly generated data symbols. Our objective is to jointly detect the transmitted data symbols and estimate the channel gains, delays, and Doppler shifts, thereby enabling simultaneous target sensing and symbol demodulation at the receiver.

\subsection{Literature Review}

Existing approaches to bistatic OTFS-ISAC can be broadly divided into two categories according to how the target parameters are estimated: two-step ISAC, which first estimates the effective channel matrix and then recover the target parameters from it, and direct parameter estimation, which estimate the physical channel parameters directly from the received DD-domain symbols.

\subsubsection{Two-Step ISAC}

A representative effective-channel estimation method was proposed in \cite{Raviteja2019}, where a high-power pilot is embedded in the DD domain and the channel taps are identified by thresholding its received response. The estimated effective channel matrix can then be further processed using spectral estimation methods to extract the delay and Doppler parameters for target sensing \cite{10251770}. However, this two-step approach relies on the shift-invariant channel structure of multi-carrier OTFS with ideal bi-orthogonal pulses \cite{OTFS0,Saif_Predictability}, under which the complete effective channel matrix can be inferred from the response of a single pilot. Such ideal pulses are not physically realizable \cite{OTFS0}. With practical pulses, such as rectangular pulses, the effective channel coefficients depend on both the DD-grid locations and the actual path parameters, and therefore vary across transmitted symbols \cite{Raviteja2018,Saif_Predictability}. Consequently, the pilot response alone is insufficient to reconstruct the complete effective channel matrix, leading to degraded data detection and inaccurate parameter extraction. This issue is particularly important in bistatic ISAC, where symbol detection errors act as additional interference and further degrade channel and target parameter estimation.

To overcome the limitations of multi-carrier OTFS, \cite{Saif_Predictability} introduced Zak-OTFS, in which the inverse Zak transform maps DD-domain symbols to the time domain and the Zak transform recovers the received DD-domain signal. In Zak-OTFS, modulation, channel propagation, and demodulation can all be represented by twisted convolution. Owing to its associativity, the resulting I/O relation reduces to a single twisted convolution between the transmitted DD-domain symbols and an effective channel response. Thus, all DD-domain symbols experience the same effective channel, even with practical pulses, allowing the complete channel matrix to be inferred from the response of a single pilot. This approach is referred to as model-free channel estimation in \cite{Saif_Predictability}. The results therein show that, with an appropriately designed frame structure, it can provide accurate channel estimates for Zak-OTFS.

Nevertheless, the effectiveness of this approach depends on how much of the effective channel response is captured by the finite DD-domain pilot observation region. When the channel response spreads over a wider range, a larger portion of its energy falls outside the observed samples, resulting in channel truncation errors. The work in \cite{mehrotra2025pulse} studied the channel estimation and symbol detection performance of Zak-OTFS under different DD-domain pulse-shaping filters. It showed that pulses with lower ambiguity-function sidelobes improve channel estimation, while pulses satisfying the Nyquist ISI criterion also facilitate reliable symbol detection. However, the effective channel response generally has infinite support in both delay and Doppler under finite time-frequency resources. Therefore, it cannot be fully reconstructed from a finite number of pilot observations, and the resulting channel approximation also limits the accuracy of subsequent DD-domain parameter extraction. This motivates the direct estimation of the physical delay-Doppler parameters from the receive signal.

\subsubsection{Parameters Estimation}

In OTFS systems, on-grid delay and Doppler parameters can be readily estimated from the shifts of the pilot response \cite{Yi_Hong_EmbeddedPilot_OTFS}. In practical ISAC scenarios, however, target delays and Doppler shifts are generally off-grid and appear as fractional multiples of the corresponding resolutions. These fractional shifts spread the response across multiple DD-domain bins, making accurate parameter estimation considerably more difficult.

Various channel estimation methods have been developed for OTFS systems with fractional DD shifts. In \cite{ICC_Fractional_OTFS}, a cross-correlation-based method was proposed to estimate the channel from the DD-domain pilot response and mitigate inter-Doppler interference. The work in \cite{wang2022joint} used variational sparse Bayesian learning (SBL) to jointly estimate the channel and detect data symbols, treating the unknown data as virtual pilots. However, it considered only integer delay and Doppler shifts. To handle off-grid channels, \cite{Zhiqiang_Wei_OTFS_SBL} developed an off-grid SBL method, but its derivation relies on ideal bi-orthogonal pulses. For practical rectangular pulses, where delay and Doppler cannot be decoupled, \cite{Yaru_Shan_OTFS_SBL} proposed a grid-evolution SBL method that iteratively refines a nonuniform two-dimensional grid. A different approach was adopted in \cite{Yongzhi_WU_OTFS}, which introduced DFT-spread OTFS with superimposed pilots for terahertz ISAC and estimated the fractional DD parameters through coarse on-grid search followed by off-grid refinement.

Although these methods provide effective DD-domain parameter estimates for OTFS sensing, they still have several limitations. The SBL-based approaches in \cite{Zhiqiang_Wei_OTFS_SBL,Yaru_Shan_OTFS_SBL} rely on discretized delay-Doppler grids, so their estimation accuracy depends on the grid resolution and improves only at the cost of higher computational complexity. The same grid mismatch issue also affects the methods in \cite{ICC_Fractional_OTFS,Yongzhi_WU_OTFS}. In addition, these works adopt restrictive assumptions, such as integer delay taps in \cite{ICC_Fractional_OTFS} and prior knowledge of the number of channel paths in \cite{Zhiqiang_Wei_OTFS_SBL,Yaru_Shan_OTFS_SBL,Yongzhi_WU_OTFS}. Such assumptions are difficult to satisfy in practical high-resolution sensing systems. To the best of our knowledge, no existing method jointly performs off-grid channel parameter estimation and data detection without relying on these assumptions. This motivates the development of a unified framework for reliable OTFS-ISAC over doubly dispersive channels.

\subsection{Contributions}

In this paper, we consider performing ISAC with bistatic sensing based on Zak-OTFS. Specifically, our main contributions are summarized as follows:

\begin{itemize}
\item \textit{Discrete-Time I/O Model for Zak-OTFS:} We derive a discrete-time matrix-form I/O relation for Zak-OTFS with rectangular time- and frequency-domain windows under fractional delay and Doppler shifts. The derived model preserves the predictability property of Zak-OTFS, showing that the model-free channel estimation method remains applicable and can provide an effective initialization for the proposed receiver.

\item \textit{Novel Problem Formulation for Semi-Blind Integrated Sensing and Communication:} Based on the derived I/O relation, we establish a unified optimization framework that jointly recovers off-grid channel parameters and unknown communication symbols from the received Zak-OTFS signal. The formulation combines atomic norm regularization for gridless delay-Doppler parameter estimation with a negative square penalty that embeds discrete PSK constellation structures into the optimization process, enabling simultaneous sensing and data detection.

\item \textit{Accelerated Algorithm and Convergence Analysis:} We develop an accelerated iterative algorithm that combines majorization-minimization, accelerated projected gradient, and inexact accelerated proximal gradient updates. We establish convergence to an $\varepsilon$-stationary point with a sublinear stationarity rate. Simulation results show that the proposed method achieves super-resolution sensing, reliable multi-target detection, and superior BER performance.
\end{itemize}

The remainder of this paper is organized as follows. Section \ref{sec2:signal_model} derives the discrete-time I/O relation of Zak-OTFS, which serves as the basis for the atomic norm denoising formulation in Section \ref{sec3:problem_formulation}. Section \ref{sec4:algorithm} details the proposed accelerated algorithm and provides a rigorous convergence analysis. Finally, simulation results and conclusions are presented in Sections \ref{sec5:numerical_results} and \ref{sec6:conclusion}, respectively.

\textit{Notations:} We use $x$, $\mathbf{x}$, $\mathbf{X}$, and $\mathcal{X}$ to represent scalar, column vector, matrix, and set, respectively. The notation $\operatorname{vec}(\cdot)$ represents the vectorization of a matrix by stacking its columns. \(\operatorname{Diag}\{\mathbf a\}\) denotes the diagonal matrix formed by placing the elements of \(\mathbf a\) on its main diagonal, while \(\operatorname{diag}\{\mathbf A\}\) denotes the vector formed by extracting the main diagonal entries of matrix \(\mathbf A\). $\|\cdot\|_{2}$ denoted the $\ell_2$-norm of a matrix. $\mathbb{C}$, $\mathbb{R}$, and $\mathbb{Z}$ denote the sets of complex, real, and integer numbers, respectively; $\langle \mathbf{x}, \mathbf{y}\rangle$ refers to the inner product between two vectors. The superscript $(\cdot)^{*}$, $(\cdot)^{\text{T}}$ and $(\cdot)^{\text{H}}$ represent the conjugate, the transpose and the conjugate transpose operations, respectively. $\mathbb{I}_{\mathcal{A}}(A)$ denotes an indicator function of $\mathcal{A}$. $\otimes$ represents the Kronecker product. $A[i,j]$ represents the element in the $i$-th row and $j$-th column of matrix $\mathbf{A}$. $\mathbb{N}_{n}$ denotes the set $\{0, 1, \dots, n-1\}$.

\section{Zak-OTFS Signal Model}\label{sec2:signal_model}

This section establishes the Zak-OTFS signal model. We first describe the ISAC scenario considered in this paper and then review the general framework of Zak-OTFS and its OFDM-based implementation. Subsequently, we derive the discrete-time matrix I/O relation and demonstrate that the predictability property remains valid, thereby enabling the model-free channel estimation scheme in \cite{Saif_Predictability}.

\begin{figure}
        \centering
        \includegraphics[width=0.95\columnwidth]{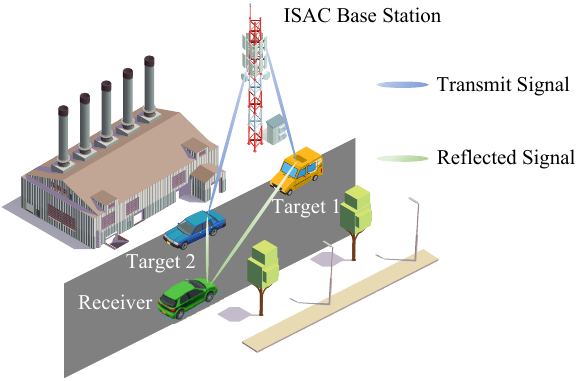}
        \caption{An illustration of the ISAC scenario.}\label{sec2_fig:scenario}
\end{figure}

\subsection{General Zak-OTFS I/O Relation}\label{sec2_1_general_DD_IO}
As illustrated in Fig. \ref{sec2_fig:scenario}, we consider a bistatic ISAC scenario, where an ISAC base station transmits the Zak-OTFS signal. The transmitted signal is reflected by multiple moving targets and subsequently collected by a spatially separated receiver. Owing to the bistatic geometry, the delay of each reflected path is determined by the total transmitter-target-receiver propagation distance, while its Doppler shift depends on the target motion relative to both the transmitter and receiver. Based on the received echoes, the receiver detects the targets and estimates their delays, Doppler shifts, and complex channel gains.

For Zak-OTFS \cite{Saif_Predictability}, we first construct the DD domain transmit signal by summing the multiplications between DD domain transmit symbols and the unfiltered basis functions, $\mathcal{Z}_{p^{\tau_l, \nu_k}}(\tau, \nu)$, i.e.,
\begin{equation}\label{sec2_eq:dd_transmit_signal}
        x_{\operatorname{dd}}(\tau, \nu) = \sum_{l=0}^{M-1} \sum_{k=0}^{N-1} X_{\operatorname{dd}}[l, k] \mathcal{Z}_{p^{\tau_l, \nu_k}}(\tau, \nu),
\end{equation}
where $X_{\operatorname{dd}}[l, k]$ is the information symbols at DD domain grid $(l, k)$, $\tau_l=\frac{l}{M \Delta f}$, $\nu_k=\frac{k}{NT}$, $T$ and $\Delta f$ satisfying $T \Delta f = 1$ are the delay and Doppler periods, respectively, and the unfiltered basis function is defined as
\begin{multline}
        \mathcal{Z}_{p^{\tau_l, \nu_k}}(\tau, \nu)   =  \sum_{m \in \mathbb{Z}} \sum_{n \in \mathbb{Z}}                                                                                                             \\
        e^{j 2 \pi \nu_k n T} \delta\left(\tau-\tau_l-n T\right) \delta\left(\nu-\nu_k-m \Delta f\right).
\end{multline}
By implementing the inverse Zak-transform\footnote{The inverse Zak-transform is defined as $s(t) \overset{\Delta}{=} \sqrt{T} \int_{0}^{\Delta f} \mathcal{Z}_{s}(t, \nu) \,d\nu$, see \cite{Saif_Predictability} for more details.}, the corresponding time domain transmit signal is given by
\begin{equation}\label{sec2_eq:td_transmit_signal}
        x(t) = \sum_{l=0}^{M-1} \sum_{k=0}^{N-1} X_{\operatorname{dd}}[l, k] p^{\tau_l, \nu_k}(t),
\end{equation}
where $p^{\tau_l, \nu_k}(t)$ is the corresponding time domain unfiltered basis signal,
\begin{equation}
        p^{\tau_l, \nu_k}(t) = \sqrt{T} \sum_{n \in \mathbb{Z}} e^{j 2 \pi \nu_k n T} \delta(t - \tau_l - nT).
\end{equation}
Since the signal \eqref{sec2_eq:td_transmit_signal} occupies infinite time and frequency domain resources, it is not physically realizable. Both the time and frequency domain resources need to be restricted to implement the Zak-OTFS in practice.

According to \cite{Saif_Predictability, Hanly_Transmitter}, the occupied time and frequency resources can be limited by performing a twisted convolution\footnote{Consider $a(\tau, \nu)$ and $b(\tau, \nu)$ are function defined on $\mathbb{R}^{2}$, the twisted convolution between them is defined as $a *_{\sigma} b (\tau, \nu) = \iint_{\mathbb{R}^{2}} a(\tau^{\prime}, \nu^{\prime})b(\tau - \tau^{\prime}, \nu - \nu^{\prime}) e^{j 2 \pi \nu^{\prime} (\tau - \tau^{\prime})} \, d\tau^{\prime} \, d \nu^{\prime}$.} with a DD domain transmit filter, $g(\tau, \nu)$, on the DD domain transmit signal in \eqref{sec2_eq:dd_transmit_signal}. The DD domain twisted convolution filtered signal is given by
\begin{equation}\label{sec2_eq:dd_filtered_transmit_signal}
        \mathcal{Z}_{x_{g}}(\tau, \nu) = g *_{\sigma} \mathcal{Z}_{x}(\tau, \nu),
\end{equation}
By implementing the inverse Zak-transform on \eqref{sec2_eq:dd_filtered_transmit_signal}, we obtain the corresponding time domain signal is denoted as $x_{g}(t)$. The time domain receive signal is
\begin{equation}\label{sec2_eq:td_receive_signal}
        r(t)=\iint h(\tau, \nu) x_{g}(t-\tau) e^{j 2 \pi \nu(t-\tau)} d \tau d \nu + n(t),
\end{equation}
where $h(\tau, \nu)$ is the doubly-dispersive channel and $n(t)$ is the time domain noise. By performing Zak-transform\footnote{For $T > 0$, the Zak-transform of a continuous time signal $s(t)$ is defined as $\mathcal{Z}_{s}(\tau, \nu) \overset{\Delta}{=} \sqrt{T} \sum_{k=-\infty}^{\infty} s(\tau + kT) e^{-j 2 \pi k \nu T}$. See \cite{Saif_Predictability} for more details.} on \eqref{sec2_eq:td_receive_signal}, the corresponding DD domain receive signal can be represented as \cite{Saif_Predictability,Hanly_OTFS_Receiver}
\begin{equation}\label{sec2_eq:dd_receive_signal}
        \mathcal{Z}_{r}(\tau, \nu) = h *_\sigma \mathcal{Z}_{x_g}(\tau, \nu) = h *_\sigma\left(g *_\sigma x_{\operatorname{dd}}\right)(\tau, \nu).
\end{equation}
The receiver side performs another twisted convolution with a DD domain receiving filter, $\tilde{g}(\tau, \nu)$, on the DD domain receive signal \eqref{sec2_eq:dd_receive_signal}, and gets the DD domain filtered signal,
\begin{equation}\label{sec2_eq:receive_twisted_conv}
        y_{\operatorname{dd}}(\tau, \nu) = \tilde{g} *_{\sigma} \mathcal{Z}_{r}(\tau, \nu).
\end{equation}

Put these results together, the DD domain I/O relation of Zak-OTFS is given by
\begin{equation}
        y_{\operatorname{dd}}(\tau, \nu) = h_{\operatorname{eff}}(\tau, \nu) *_{\sigma} x_{\operatorname{dd}}(\tau, \nu) + \tilde{n}_{\operatorname{dd}}(\tau, \nu),
\end{equation}
where $h_{\operatorname{eff}}(\tau, \nu) = \tilde{g} *_{\sigma} h *_{\sigma} g(\tau, \nu)$ is the effective channel response, and $\tilde{n}_{\operatorname{dd}}(\tau, \nu) = \tilde{g} *_{\sigma} \mathcal{Z}_{n}(\tau, \nu)$ is the DD domain effective noise. By sampling $y_{\operatorname{dd}}(\tau, \nu)$ at $\tau = \tau_l$ and $\nu = \nu_k$ and omitting the noise term, we have
\begin{equation}\label{sec2_eq:DD_discret_IO}
        Y_{\operatorname{dd}}[l, k] = \sum_{l^{\prime} \in \mathbb{Z}} \sum_{k^{\prime} \in \mathbb{Z}} h_{\operatorname{eff}}[l^{\prime}, k^{\prime}] X_{\operatorname{dd}}[l-l^{\prime}, k-k^{\prime}] e^{j 2 \pi \frac{k^{\prime} (l - l^{\prime})}{MN}},
\end{equation}
where $Y_{\operatorname{dd}}[l, k]$ is the DD domain received symbols at $(l, k)$, $h_{\operatorname{eff}}[l, k]$ and $X_{\operatorname{dd}}[l, k]$ means sampling $h_{\operatorname{eff}}(\tau, \nu)$ and $x_{\operatorname{dd}}(\tau, \nu)$ at $\tau = \tau_l$ and $\nu = \nu_k$, respectively. The equation \eqref{sec2_eq:DD_discret_IO} is called the discrete twisted convolution by \cite{Saif_Predictability}.

\subsection{Specific Zak-OTFS Implementation via OFDM Framework}\label{sec2_2_DD_IO_OFDM_based}
In this subsection, we consider a specific implementation of Zak-OTFS and introduce how to realize it under the OFDM framework. The DD domain transmit filter in this specific implementation is constructed as
\begin{equation}\label{sec2_eq:dd_filter}
        g(\tau, \nu) = \alpha(\tau) \beta(\nu) e^{j 2 \pi \tau \nu}.
\end{equation}
According to \cite{Hanly_Transmitter}, performing the twisted convolution \eqref{sec2_eq:dd_filtered_transmit_signal} with $g(\tau, \nu)$ given in \eqref{sec2_eq:dd_filter} is equivalent to implement frequency domain windowing first and then time domain windowing on \eqref{sec2_eq:td_transmit_signal}. The corresponding filtered basis function is expressed as
\begin{equation}
        p^{\tau_l, \nu_k}_{g}(t) = \underbrace{B(t) e^{j 2 \pi \nu_k (t - \tau_l)}}_{\text{Time-limited tone}} \underbrace{\sum_{m \in \mathbb{Z}} A(m \Delta f + \nu_k) e^{j 2 \pi m \Delta f (t - \tau_l)}}_{\text{Pulse train}},
\end{equation}
where $A(f) = \int_{\mathbb{R}} \alpha(\tau) e^{-j 2 \pi f \tau} \, d\tau$ is the frequency domain window function, $B(t) = \int_{\mathbb{R}} \beta(\nu) e^{j 2 \pi \nu t} \, d\nu$ is the time domain window function, $p^{\tau_l, \nu_k}_{g}(t)$ is known as \textit{pulsone} \cite{Saif_Predictability}.

\begin{figure*}[t]
        \centering
        \includegraphics[width=0.8\textwidth]{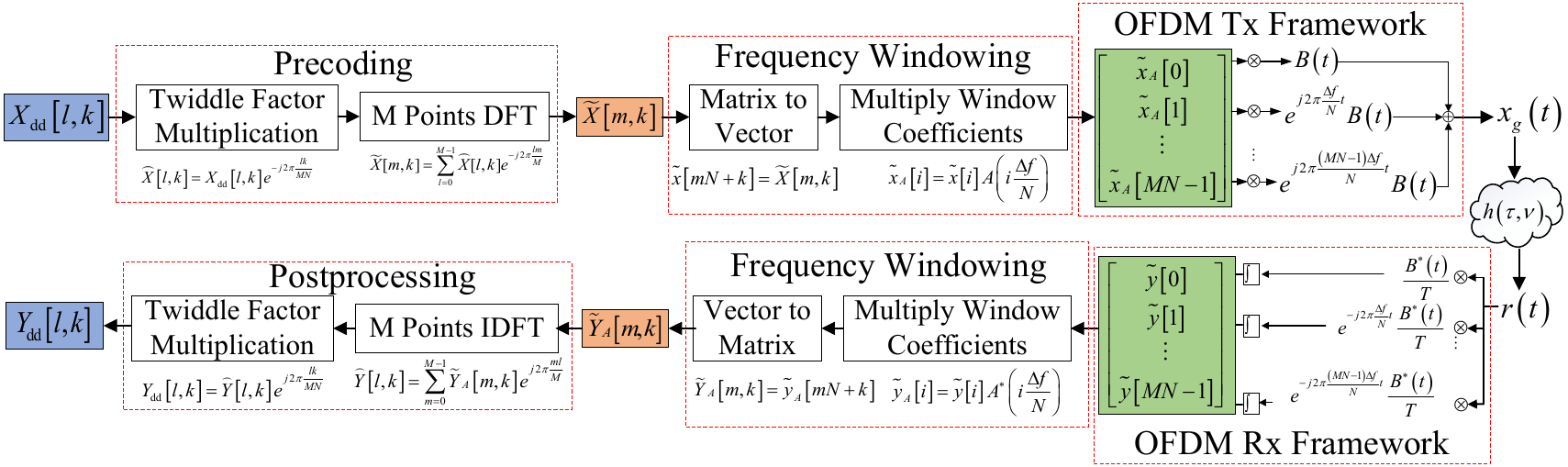}
        \caption{OFDM-based transceiver framework of windowed Zak-OTFS \cite{Hanly_Transmitter, Zak_OTFS_Receiver}.}
        \label{fig:ZakOTFS_OFDM_framework}
\end{figure*}

We assume the channel is sparse and there are $P$ resolvable paths, which gives
\begin{equation}
        h(\tau, \nu) = \sum_{i=1}^{P} h_{i} \delta(\tau - \tau_{i}) \delta(\nu - \nu_{i}),
\end{equation}
where $\tau_i$ and $\nu_i$ are the delay and Doppler shift of the $i$-th path, respectively. Then, the time-domain receive signal becomes
\begin{equation}
        r(t) = \sum_{i=1}^{P} h_{i} x_{g}(t - \tau_{i}) e^{j 2 \pi \nu_{i} (t - \tau_{i})} + n(t).
\end{equation}
The DD domain receive filter of this specific implementation is constructed as
\begin{equation}\label{sec2_eq:dd_receive_filter}
        \tilde{g}(\tau, \nu) = \alpha^{*}(-\tau) \beta^{*}(-\nu).
\end{equation}
Performing the twisted convolution \eqref{sec2_eq:receive_twisted_conv} with $\tilde{g}(\tau, \nu)$ defined above is equivalent to time windowing the receive signal $r(t)$ with $B^{*}(t)/T$ followed by frequency windowing it with $A^{*}(f)$.

The transceiver procedure of Zak-OTFS defined in this subsection can be implemented via the OFDM framework, as shown in Fig. \ref{fig:ZakOTFS_OFDM_framework}. It can be verified that the DD domain I/O relation under the transceiver frameworks shown in Fig. \ref{fig:ZakOTFS_OFDM_framework} is exactly the same as \eqref{sec2_eq:DD_discret_IO}.

\subsection{Matrix Form I/O Relation of Zak-OTFS}\label{sec2:matrix_IO}

The aforementioned continuous-time I/O relation provides mathematical insights but it is inconvenient for digital implementation. Although some simulation-oriented works truncate the infinite DD-domain effective channel response \cite{mehrotra2026pulse}, they lack an error characterization of the truncation approximation. To bridge this gap and facilitate computer simulations, we derive a discrete-time counterpart for the Zak-OTFS I/O relationship.

As illustrated in Fig. \ref{fig:ZakOTFS_OFDM_framework}, the receive framework processes the receive signal $r(t)$ through three cascading stages: OFDM demodulation, frequency windowing, and postprocessing. For the discrete-time Zak-OTFS, the receive signal is sampled at a rate of $F_s = M \Delta f$. The intermediate vector $\tilde{\mathbf{y}}$ is then obtained by computing the discrete-time cross-correlation between the sampled receive signal and the $n$-th windowed subcarrier $[B^{*}(k/F_s)/Te^{-j 2 \pi \frac{n k \Delta f}{N F_s}}]_{k=0}^{MN-1}$. Then, we perform the subsequent procedures to obtain the DD domain receive symbols.

Without loss of generality, we consider that both the time and frequency windows are rectangular functions, i.e., $A(f)$ and $B(t)$ are $1$ within the interval $[0, M\Delta f]$ and $[0, NT]$, respectively, and zero, otherwise. Assume the cyclic-prefix (CP) is added in front of the transmit signal $x_{g}(t)$. After removing the CP at the receiver side, the $k$-th sample of the time domain receive signal is
\begin{equation}\label{sec2_eq:td_receive_signal_sampled}
        r[k] = \sum_{i=1}^{P} \sum_{n = 0}^{MN - 1} h_{i} s_{n} e^{j 2 \pi \left(\frac{n}{NT} + \nu_{i}\right) \left(\frac{k}{M\Delta f} - \tau_{i}\right)} + n[k],
\end{equation}
where $s_{n}$ is the $n$-th element of the transmit symbols sequence, $\mathbf{s}$, given by
\begin{equation}\label{sec2_eq:Zak_OTFS_modulation}
        \mathbf{s} = \operatorname{vec}\{(\mathbf{X}_{\operatorname{dd}} \cdot \mathbf{W})^{\operatorname{T}}\mathbf{F}_{M}\} = (\mathbf{F}_{M} \otimes \mathbf{I}_{N}) \widetilde{\mathbf{W}} \mathbf{x}_{\operatorname{dd}},
\end{equation}
$\mathbf{X}_{\operatorname{dd}} \in \mathbb{C}^{M \times N}$ denote the DD domain symbol matrix, $\mathbf{x}_{\operatorname{dd}} = \operatorname{vec}\{\mathbf{X}_{\operatorname{dd}}^{\operatorname{T}}\}$ is its vectorized form, $\mathbf{W} \in \mathbb{C}^{M \times N}$ is the twiddle factor matrix, $\widetilde{\mathbf{W}} = \operatorname{Diag}\{\operatorname{vec}\{\mathbf{W}^{\operatorname{T}}\}\}$, and $\mathbf{F}_{M} \in \mathbb{C}^{M \times M}$ is the discrete Fourier transform matrix satisfying $\mathbf{F}_{M} \mathbf{F}_{M}^{\operatorname{H}} = \mathbf{I}_{M}$. Based on \eqref{sec2_eq:td_receive_signal_sampled}, the time domain receive signal can be rewritten into vector form,
\begin{equation}\label{sec2_eq:td_vec_receive_signal}
        \mathbf{r} = \sum_{i=1}^{P} h_{i} \mathbf{D}_{\nu_{i}} \mathbf{F}_{MN}^{\operatorname{H}} \mathbf{B}_{\tau_{i}} \mathbf{s} + \mathbf{n},
\end{equation}
where $\mathbf{B}_{\tau} = \operatorname{Diag}\left\{\left[1, e^{-j2\pi\frac{\tau}{NT}}, \dots, e^{-j2\pi\frac{(MN-1)\tau}{NT}}\right]\right\}$, $\mathbf{D}_{\nu} = \operatorname{Diag}\left\{\left[1, e^{j2\pi\frac{\nu}{M\Delta f}}, \dots, e^{j2\pi\frac{(MN-1)\nu}{M\Delta f}}\right]\right\}$, and $\mathbf{n} \in \mathbb{C}^{MN \times 1}$ denotes the receiver noise vector, which follows a complex Gaussian distribution, i.e.,
$\mathbf{n} \sim \mathcal{CN}\!\left(\mathbf{0}, \sigma_n^2 \mathbf{I}_{MN}\right)$.

Based on the receive framework, the DD domain receive signal can also be written in vector form,
\begin{equation}\label{sec2_eq:DD_receive_vec}
        \mathbf{y}_{\operatorname{dd}} = \frac{1}{T} \widetilde{\mathbf{W}}^{\operatorname{H}} (\mathbf{F}_{M}^{\operatorname{H}} \otimes \mathbf{I}_{N}) \mathbf{F}_{MN} \mathbf{r},
\end{equation}
where the DD domain receive matrix $\mathbf{Y}_{\operatorname{dd}}$ in Fig. \ref{fig:ZakOTFS_OFDM_framework} is obtained by reshaping $\mathbf{y}_{\operatorname{dd}}$ back into matrix form by constructing the columns one-by-one. The I/O relation \eqref{sec2_eq:DD_receive_vec} can be rewritten as follows
\begin{equation}\label{sec2_eq:matrix_DD_IO}
        \mathbf{y}_{\operatorname{dd}} = \mathbf{H}_{\operatorname{eff}} \mathbf{x}_{\operatorname{dd}} + \tilde{\mathbf{n}},
\end{equation}
where $\tilde{\mathbf{n}} = \widetilde{\mathbf{W}}^{\operatorname{H}} (\mathbf{F}_{M}^{\operatorname{H}} \otimes \mathbf{I}_{N}) \mathbf{F}_{MN} \mathbf{n}$ is the DD domain effective noise vector, and $\mathbf{H}_{\operatorname{eff}}$ is the DD domain effective channel matrix given by
\begin{equation}\label{sec2_eq:eff_DD_channel_mtx}
        \mathbf{H}_{\operatorname{eff}} = \frac{1}{T} \widetilde{\mathbf{W}}^{\operatorname{H}} (\mathbf{F}_{M}^{\operatorname{H}} \otimes \mathbf{I}_{N}) \mathbf{F}_{MN} \mathbf{H}_{\operatorname{TD}} \mathbf{F}_{MN}^{\operatorname{H}} (\mathbf{F}_{M} \otimes \mathbf{I}_{N}) \widetilde{\mathbf{W}},
\end{equation}
where $\mathbf{H}_{\operatorname{TD}} = \sum_{i=1}^{P} h_{i} \mathbf{D}_{\nu_{i}} \mathbf{F}_{MN}^{\operatorname{H}} \mathbf{B}_{\tau_{i}} \mathbf{F}_{MN}$ is the time domain channel matrix. The $(k^{\prime}M + l^{\prime}, kM + l)$-th element in $\mathbf{H}_{\operatorname{eff}}$ is given by
\begin{multline}\label{sec2_eq:effective_channel_matrix}
        \hspace{-1em}H_{\operatorname{eff}}[k^{\prime}M + l^{\prime}, kM + l] = \frac{1}{T} \sum_{i=1}^{P} \sum_{m=0}^{M-1} \sum_{m^{\prime} = 0}^{M-1} h_{i} e^{j 2 \pi \frac{(l^{\prime} - l - l_{\tau_{i}})(k + m N)}{MN}}\\
        \sum_{n=0}^{MN-1} e^{-j 2 \pi n \frac{k^{\prime} - k - (m-m^{\prime}) N - k_{\nu_{i}}}{MN}}  e^{j 2 \pi \frac{l^{\prime} [k^{\prime} - k - (m-m^{\prime}) N]}{MN}}.
\end{multline}
Noticing that we have the equation
\begin{figure*}
        \begin{equation}\label{sec2_eq:freq_ambiguity_function_summation}
                \sum_{m \in \mathbb{Z}} e^{j 2 \pi \frac{k m}{N}} \mathcal{Y}_{A}(\tau_{l^{\prime}} - \tau_{l} - \tau - m T, -(\nu_{k^{\prime}} - \nu_{k} - \bar{m} \Delta f)) = \frac{1}{T} \sum_{m \in \mathbb{Z}} A(\nu_{k} + m \Delta f) A^{*}((m - \bar{m}) \Delta f + \nu_{k^{\prime}}) e^{j 2 \pi \frac{(l^{\prime} - l - l_{\tau}) (m N + k)}{M N}}
        \end{equation}
\end{figure*}
in \eqref{sec2_eq:freq_ambiguity_function_summation} by applying the Poisson summation, where $\mathcal{Y}_{A}(\tau, \nu) = \int_{\mathbb{R}} A(f) A^{*}(f - \nu) e^{j 2 \pi f \tau} \, df$ is the ambiguity function of the frequency window \cite{Hanly_OTFS_Receiver}. Recap that $A(f) = 1$ when $f \in [0, M \Delta f]$, and it is zero, otherwise; together with the fact that $k \in \mathbb{N}_{N}$ and $k^{\prime} \in \mathbb{N}_{N}$, the value of \eqref{sec2_eq:freq_ambiguity_function_summation} is nonzero only when $m \in \mathbb{N}_{M}$ and $(m - \bar{m}) \in \mathbb{N}_{M}$. By denoting $\bar{m} = m - m^{\prime}$, and substituting \eqref{sec2_eq:freq_ambiguity_function_summation} into \eqref{sec2_eq:effective_channel_matrix}, the effective channel matrix can be rewritten as
\begin{multline}\label{sec2_eq:element_wise_eff_ch}
        H_{\operatorname{eff}}[k^{\prime}M + l^{\prime}, kM + l] = \sum_{m \in \mathbb{Z}} \sum_{\bar{m} \in \mathbb{Z}} e^{j 2 \pi \frac{(k^{\prime} - k - \bar{m} N)(l + m M)}{MN}}\\
        e^{j 2 \pi \frac{m k}{N}} h_{\operatorname{eff}}[l^{\prime} - l - m M, k^{\prime} - k - \bar{m} N],
\end{multline}
where $h_{\operatorname{eff}}[l, k]$ is the effective channel response given by
\begin{equation}\label{sec2_eq:h_eff_specific}
        h_{\operatorname{eff}}[l, k] = \sum_{i=1}^{P} h_{i} e^{j 2 \pi \frac{lk}{MN}} \mathcal{Y}_{A}(\tau_{l} - \tau_{i}, -\nu_{k}) \mathcal{X}_{\tilde{B}}(0, k - k_{\nu_{i}}),
\end{equation}
with $\mathcal{X}_{\tilde{B}}(0, k_{\nu}) = \sum_{n=0}^{MN-1} e^{-j 2 \pi \frac{n k_{\nu}}{MN}}$ being the zero delay-cut on the discrete periodic ambiguity function \cite{benedetto2009phase} of the time domain window samples, i.e., $[B(n / F_s)]_{n=0}^{MN-1}$. Combine \eqref{sec2_eq:matrix_DD_IO} and \eqref{sec2_eq:element_wise_eff_ch} together, the element-wise I/O relation for the discretization form of the framework in Fig. \ref{fig:ZakOTFS_OFDM_framework} can also be expressed as the discretized twisted convolution as in \eqref{sec2_eq:DD_discret_IO} with $h_{\operatorname{eff}}[l, k]$ given in \eqref{sec2_eq:h_eff_specific}, meaning that the predictability property \cite{Saif_Predictability} remains valid for the discrete-time Zak-OTFS system.

\subsection{Model-Free Channel Estimation with Embedded Pilot}\label{sec2:embedded_pilot_scheme}

\begin{figure}[t]
        \centering
        \includegraphics[width=0.8\columnwidth]{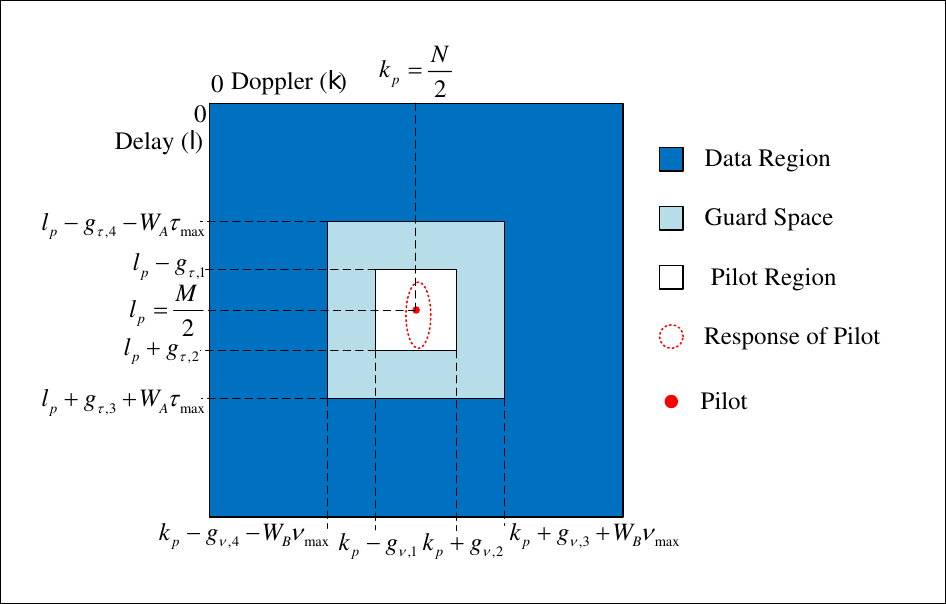}
        \caption{Zak-OTFS frame with embedded pilot.}\label{sec2_fig:embedded_pilot_illustration}
\end{figure}

The Zak-OTFS frame structure with the embedded pilot scheme is shown in Fig. \ref{sec2_fig:embedded_pilot_illustration}. The sizes of guard space and pilot region of the Zak-OTFS frame are determined according to the channel characters, i.e., the maximum delay, $\tau_{\max}$, and the maximum Doppler, $\nu_{\max}$. By denoting the pilot symbol as $x_{p}$, which is placed at $(l_p, k_p) = \left(\lfloor \frac{M}{2} \rfloor, \lfloor \frac{N}{2} \rfloor\right)$, the received pilot response is
\begin{multline}\label{sec2_eq:pilot_response}
        y_{\operatorname{dd}, p}[l, k] =  x_p h_{\operatorname{eff}}[l - l_p, k - k_p] e^{j 2 \pi \frac{(k - k_p)l_p}{MN}} \\
        + \sum_{m \in \mathbb{Z} \backslash \{0\}} \sum_{n \in \mathbb{Z} \backslash \{0\}} x_p h_{\operatorname{eff}}[l - l_p - nM, k - k_p - mN] \\
        e^{j 2 \pi \frac{nk}{N}} e^{j 2 \pi \frac{(k - k_p - mN)(l_p + nM)}{MN}}.
\end{multline}
Under the crystalline regime condition \cite{Saif_Predictability}, which holds when $\tau_{\max} < T$, $\nu_{\max} < \Delta f$, and $M$ and $N$ are large enough, the second term in \eqref{sec2_eq:pilot_response} can be ignored so that the effective channel response can be approximated by sampling the response in pilot region, i.e.,
\begin{equation}\label{sec2_channel_estimation_scheme}
        \hat{h}_{\operatorname{eff}}[l, k] = \begin{cases}
                \frac{y_{\operatorname{dd}}[l+l_p, k+k_p]}{x_p} e^{-j 2\pi \frac{k l_p}{MN}}, \text{if } (l, k) \text{ in pilot region}, \\
                0,                                        \text{otherwise}.
        \end{cases}
\end{equation}
The estimated effective channel matrix can be reconstructed through \eqref{sec2_eq:element_wise_eff_ch} and \eqref{sec2_channel_estimation_scheme}.

\section{Problem Formulation}\label{sec3:problem_formulation}

In this section, we formulate the optimization problem for Zak-OTFS bistatic ISAC. We first formulate the joint parameter estimation and symbol detection task as an atomic norm denoising problem to exploit the sparse nature of the DD channel. Subsequently, to address the computational intractability caused by the discrete PSK constellation constraints on information symbols, we introduce a penalty-based reformulation that relaxes the discrete constraints into a convex set, setting the stage for the efficient algorithm developed in the next section.

\subsection{Atomic Norm Formulation}\label{sec3:atomic_norm_formulation}

The time-domain I/O relation derived in \eqref{sec2_eq:td_vec_receive_signal} can be re-expressed in a compact matrix form,
\begin{equation}
        \mathbf{r} = (\mathbf{s}^{\operatorname{T}} \otimes \mathbf{I}_{MN}) \operatorname{Diag}\{\operatorname{vec}\{\mathbf{F}_{MN}^{\operatorname{H}}\}\} \mathbf{h} + \mathbf{n},
\end{equation}
where the channel vector $\mathbf{h}$ is defined as $\mathbf{h} = \sum_{i=1}^{P} h_i \mathbf{b}_{\tau_i} \otimes \mathbf{d}_{\nu_i}$, with $\mathbf{b}_{\tau_i} = \operatorname{diag}\{\mathbf{B}_{\tau_i}\}$ and $\mathbf{d}_{\nu_i} = \operatorname{diag}\{\mathbf{D}_{\nu_i}\}$. To leverage the inherent sparsity of the channel in the DD domain, we employ the atomic norm \cite{Chi2020} of $\mathbf{h}$ as a regularizer. The joint parameter estimation and symbol detection task is formulated as an atomic norm denoising problem:
\begin{subequations}\label{sec3_prob:origin_atomic_norm}
        \begin{align}
                 & \min_{\mathbf{x}_{\operatorname{dd}, d}, \mathbf{h}} \frac{1}{2}\|\mathbf{r} - (\mathbf{s}^{\operatorname{T}} \otimes \mathbf{I}_{MN}) \mathbf{G} \mathbf{h}\|_{2}^{2} + \eta \|\mathbf{h}\|_{\mathcal{A}}\label{sec3_prob:origin_atomic_norm_prob} \\
                 & \;\;\;\text{s.t.}\; \mathbf{x}_{\operatorname{dd}, d} \in \mathcal{S}\label{sec3_prob:origin_atomic_norm_constraint},
        \end{align}
\end{subequations}
where $\mathbf{x}_{\operatorname{dd}, d}$ represents the vector of data information symbols located within the `data region' of the Zak-OTFS frame (which is a sub-vector of $\mathbf{x}_{\operatorname{dd}}$), $\mathbf{G} = \operatorname{Diag}\{\operatorname{vec}\{\mathbf{F}_{MN}^{\operatorname{H}}\}\}$ is a diagonal matrix, and $\eta$ is the regularization parameter. The set $\mathcal{S}$ denotes the discrete alphabet set (e.g., BPSK or QPSK). The atomic norm $\|\mathbf{h}\|_{\mathcal{A}}$ is defined as \cite{Chi2020}:
\begin{equation}\label{sec3_def:atomic_norm}
        \|\mathbf{h}\|_{\mathcal{A}} = \inf\{t > 0 \mid \mathbf{h} \in t \cdot \operatorname{conv}\{\mathcal{A}\}\},
\end{equation}
where $\mathcal{A}$ is the set of atoms \cite{Chi2020}:
\begin{multline}
        \mathcal{A} = \{e^{j \phi} \mathbf{a}(\tau, \nu) \mid \mathbf{a}(\tau, \nu) = \mathbf{b}_{\tau} \otimes \mathbf{d}_{\nu},\\
        \phi \in [0, 2\pi), (\tau, \nu) \in \varOmega\},
\end{multline}
and $\operatorname{conv}\{\mathcal{A}\}$ denotes the convex hull of $\mathcal{A}$.

Directly solving problem \eqref{sec3_prob:origin_atomic_norm} presents significant challenges. First, the definition of the atomic norm involves an infinite-dimensional convex hull, which typically requires reformulation into a semi-definite programming (SDP) problem \cite{Chi2020}. While methods like the alternating direction method of multipliers (ADMM) can solve such SDPs, they remain computationally expensive for large-scale systems. More critically, the constraint $\mathbf{x}_{\operatorname{dd}, d} \in \mathcal{S}$ imposes a discrete, non-convex feasible set, rendering \eqref{sec3_prob:origin_atomic_norm} a mixed-integer programming problem that is generally NP-hard. Therefore, a new algorithmic framework is required to efficiently solve \eqref{sec3_prob:origin_atomic_norm} while properly handling both the atomic-norm regularization and the discrete symbol constraints.

\subsection{Penalty-Based Reformulation}\label{sec3:penalty_model}

To address the intractability of the discrete alphabet constraint, we employ a penalty method that relaxes the discrete set $\mathcal{S}$ to its convex hull while enforcing structure through a negative square penalty term \cite{Shao2019, wu2024quantized}. This strategy transforms the combinatorial problem into a convex optimization problem. Specifically, it relaxes the constraint $\mathbf{x}_{\operatorname{dd}, d} \in \mathcal{S}$ to $\mathbf{x}_{\operatorname{dd}, d} \in \bar{\mathcal{S}}$, where $\bar{\mathcal{S}}=\operatorname{conv}\{\mathcal{S}\}$ is the convex hull of the constellation. Ideally, the solution should lie on the vertices of $\bar{\mathcal{S}}$, which correspond to the original discrete symbols. To encourage this, a concave penalty term $-\rho \|\mathbf{x}_{\operatorname{dd}, d}\|_{2}^{2}$ is added to the objective function \cite{Shao2019, wu2024quantized}, resulting in the following formulation:
\begin{equation}\label{sec4_prob:penalty_model_anm}
        \hspace{-0.6em}\begin{aligned}
                 & \min_{\mathbf{x}_{\operatorname{dd}, d}, \mathbf{h}} \frac{1}{2}\|\mathbf{r} - (\mathbf{s}^{\operatorname{T}} \otimes \mathbf{I}_{MN}) \mathbf{G} \mathbf{h}\|_{2}^{2} - \rho \|\mathbf{x}_{\operatorname{dd}, d}\|_{2}^{2} + \eta \|\mathbf{h}\|_{\mathcal{A}} \\
                 & \;\;\;\text{s.t.}\; \mathbf{x}_{\operatorname{dd}, d} \in \bar{\mathcal{S}},
        \end{aligned}
\end{equation}
where $\rho > 0$ is the penalty parameter. This approach is theoretically underpinned by the following lemma, which establishes the equivalence between the relaxed penalty problem \eqref{sec4_prob:penalty_model_anm} and the original problem \eqref{sec3_prob:origin_atomic_norm} for sufficiently large penalty parameters.
\begin{lemma}[\cite{wu2024quantized}]\label{sec4_lemma:penalty_parameter}
        For any fixed channel vector $\mathbf{h}$, there exists a threshold $\rho_0 > 0$ such that for all $\rho > \rho_0$, the (local) optimal solution for $\mathbf{x}_{\operatorname{dd}, d}$ in the penalty-based problem \eqref{sec4_prob:penalty_model_anm} is identical to that of the original discrete problem \eqref{sec3_prob:origin_atomic_norm}.
\end{lemma}

As explained in \cite{wu2024quantized}, for constant-modulus constellations (e.g., BPSK or QPSK), the discrete elements of $\mathcal{S}$ are the vertices of $\bar{\mathcal{S}}$ and these vertices maximize the $\ell_{2}$-norm, $\|\mathbf{x}_{\operatorname{dd}, d}\|_{2}^{2}$. By promoting a larger norm through the data vector, the negative square penalty term, i.e., $-\rho \|\mathbf{x}_{\operatorname{dd}, d}\|_{2}^{2}$, encourages the solution of \eqref{sec4_prob:penalty_model_anm} to approach a discrete constellation point as the penalty weight $\rho$ grows. However, it is worth mentioning that the reformulation \eqref{sec4_prob:penalty_model_anm} is not suitable for higher-order QAM constellations, where many constellation points lie inside the convex hull. Since the negative square penalty tends to favor boundary points rather than all valid QAM symbols, the relaxed problem is no longer equivalent to the original QAM-constrained problem. Therefore, the proposed algorithm is designed for PSK-modulated symbols, while extending it to high-order QAM is left for future work.

While penalty methods have been successfully applied to problems like quantized precoding \cite{Shao2019, wu2024quantized, wu2023efficient}, our current formulation \eqref{sec4_prob:penalty_model_anm} introduces distinct complexities. Unlike prior works that optimize a single variable, our objective function features a bilinear coupling between the optimization variables $\mathbf{x}_{\operatorname{dd}, d}$ and $\mathbf{h}$. This coupling creates a neither convex nor concave objective function with potential saddle points, complicating the convergence analysis. Furthermore, the inclusion of the atomic norm regularizer $\|\mathbf{h}\|_{\mathcal{A}}$ precludes simple gradient-based updates for $\mathbf{h}$. Addressing these challenges requires a tailored algorithmic framework capable of handling the biconvex structure and the atomic norm regularization efficiently, which we develop in the next section.

\section{Iterative Coordinate Descent Method for Zak-OTFS Bistatic ISAC}\label{sec4:algorithm}

In this section, we present an efficient iterative algorithm for joint parameter estimation and symbol detection in Zak-OTFS ISAC systems. Then, we provide a theoretical convergence analysis establishing the algorithm's convergence to an $\varepsilon$-stationary point.

\subsection{Our Proposed Algorithm}
In this subsection, we propose an iterative algorithm for solving \eqref{sec4_prob:penalty_model_anm}. To update the channel vector, we employ an inexact accelerated proximal gradient method to solve the atomic norm denoising subproblem. For the symbol update, we utilize an accelerated projected gradient method.

\subsubsection{Initialization}\label{sec4:initialization}
We initialize the algorithm by estimating the channel through \eqref{sec2_channel_estimation_scheme}. By denoting the estimated DD domain channel matrix as $\hat{\mathbf{H}}_{\operatorname{eff}}$, the DD domain transmit symbol vector is detected by LMMSE detector:
\begin{equation}\label{sec2_eq:LMMSE_detector}
        \hat{\mathbf{x}}_{\operatorname{dd}} = \left(\hat{\mathbf{H}}_{\operatorname{eff}}^{\operatorname{H}} \hat{\mathbf{H}}_{\operatorname{eff}} + \sigma^2 \mathbf{I}_{MN}\right)^{-1} \hat{\mathbf{H}}_{\operatorname{eff}}^{\operatorname{H}} \mathbf{y}_{\operatorname{dd}}.
\end{equation}
Then we get the initial channel vector for problem \eqref{sec4_prob:penalty_model_anm} through $\mathbf{h}^{(0)} = \operatorname{vec}\{\mathbf{H}^{(0)}\}$, where $\mathbf{H}^{(0)}$ is given by
\begin{equation}
        \mathbf{H}^{(0)} = \mathbf{F}_{MN} \odot\left(\mathbf{Q} \hat{\mathbf{H}}_{\operatorname{eff}} \mathbf{Q}^{\operatorname{H}} \mathbf{F}_{MN}^{\operatorname{H}}\right), \mathbf{Q} = \mathbf{F}_{MN}^{\operatorname{H}} (\mathbf{F}_{M} \otimes \mathbf{I}_{N}) \widetilde{\mathbf{W}}.\nonumber
\end{equation}
We then obtain the initial time domain signal vector $\mathbf{s}^{(0)}$ through \eqref{sec2_eq:Zak_OTFS_modulation} with $\hat{\mathbf{x}}_{\operatorname{dd}}$.

\subsubsection{Updating the Channel Vector}
Assume the variables have been updated $t$ times. At the $(t+1)$-th iteration, for the simplicity of notations, we reformulate the subproblem for updating $\mathbf{h}^{(t+1)}$ as below,
\begin{equation}\label{sec4_prob:AST_h}
        \min_{\mathbf{h}} \gamma_{t}(\mathbf{h}) + \eta \|\mathbf{h}\|_{\mathcal{A}},
\end{equation}
where $\gamma_{t}(\mathbf{h}) = \frac{1}{2} \|\mathbf{r} - \mathbf{S}^{(t)} \mathbf{h}\|_{2}^{2}$ with $\mathbf{S}^{(t)} = ([\mathbf{s}^{(t)}]^{\operatorname{T}} \otimes \mathbf{I}_{MN}) \mathbf{G}$. Instead of solving problem \eqref{sec4_prob:AST_h} directly, we consider implementing one-step inexact accelerated proximal-gradient method. Specifically, we choose the Nesterov or FISTA-type acceleration scheme. We update the channel vector by $\mathbf{h}^{(t+1)} = \operatorname{prox}_{\alpha^{(t)}}(\hat{\mathbf{h}}^{(t)})$, where
\begin{equation}\label{sec4_prob:AST_h_proximal}
        \operatorname{prox}_{\alpha^{(t)}}(\hat{\mathbf{h}}^{(t)}) = \argmin_{\mathbf{h}} \frac{\alpha^{(t)}}{2} \|\mathbf{h} - \hat{\mathbf{h}}^{(t)}\|_{2}^{2} + \eta \|\mathbf{h}\|_{\mathcal{A}},
\end{equation}
$\alpha^{(t)} = \lambda_{\max}([\mathbf{S}^{(t)}]^{\operatorname{H}}\mathbf{S}^{(t)}) = \|\mathbf{s}^{(t)}\|_{2}^{2}$ is the maximum eigenvalue, $\hat{\mathbf{h}}^{(t)} = \tilde{\mathbf{h}}^{(t)} - \frac{1}{\alpha^{(t)}} \nabla \gamma_{t}(\tilde{\mathbf{h}}^{(t)})$, $\nabla \gamma_{t}(\tilde{\mathbf{h}}^{(t)})$ is the gradient of $\gamma_{t}(\mathbf{h})$ at $\tilde{\mathbf{h}}^{(t)}$ with
\begin{equation}\label{sec4_eq:mid_h_acc_proximal}
        \tilde{\mathbf{h}}^{(t)} = \mathbf{h}^{(t)} + \mu^{(t)} (\mathbf{h}^{(t)} - \mathbf{h}^{(t-1)}),
\end{equation}
and
\begin{equation}
        \mu^{(t)} = \frac{\xi^{(t-1)} - 1}{\xi^{(t)}}, \xi^{(t)} = \frac{1 + \sqrt{4 [\xi^{(t-1)}]^2}}{2},
\end{equation}
and with $\xi^{(-1)} = 0$ and $\mathbf{h}^{(-1)} = \mathbf{h}^{(0)}$.

Solving \eqref{sec4_prob:AST_h} via the one-step inexact accelerated proximal-gradient method offers two key advantages. First, it ensures numerical stability by preventing oscillation within the solution subspace caused by the rank deficiency of $\mathbf{S}^{(t)}$. Second, it enhances computational efficiency by replacing expensive exact optimization on \eqref{sec4_prob:AST_h} with a one-step inexact solution towards \eqref{sec4_prob:AST_h_proximal}, which suffices for overall convergence. While \eqref{sec4_prob:AST_h_proximal} can be solved via an SDP-based ADMM approach, the computational cost is prohibitive. Instead, we leverage the coordinate descent framework from \cite{Coordinate_ANM_Li_2024}, which exploits the origin-symmetry of the atomic set to achieve an efficiency gain of approximately $10^3$ times over state-of-the-art ADMM solver.

Observing that our atomic set $\mathcal{A}$ satisfies this requisite origin-symmetry, we invoke the results from \cite{Chi2020, Coordinate_ANM_Li_2024} to express the atomic norm \eqref{sec3_def:atomic_norm} in the following equivalent form:
\begin{multline}\label{sec3_def:atomic_norm_mix_int}
        \|\mathbf{h}\|_{\mathcal{A}} = \inf_{\substack{(\tau_i, \nu_i) \in \varOmega\\\phi_{i} \in [0, 2\pi)}}\left\{\sum_{i} c_{i} | \mathbf{h} = \sum_{i} c_{i} e^{j \phi} \mathbf{a}(\tau_i, \nu_i),\right.\\
        \left.c_{i}>0, \mathbf{a}(\tau_i, \nu_i) \in \mathcal{A}\right\}.
\end{multline}
With \eqref{sec3_def:atomic_norm_mix_int}, problem \eqref{sec3_prob:origin_atomic_norm} can be reformulated into a mixed-integer problem
\begin{equation}\label{sec4_prob:atomic_mix_int}
        \min_{\substack{c_{i},\phi_{i},            \\\tau_{i}, \nu_{i}, L}} \frac{\alpha^{(t)}}{2}\|\hat{\mathbf{h}}^{(t)} - \sum_{i=1}^{L} c_{i} e^{j \phi_{i}} \mathbf{a}(\tau_i, \nu_i)\|_{2}^{2} + \eta \sum_{i=1}^{L} c_{i}.
\end{equation}
Following the procedure in \cite{Coordinate_ANM_Li_2024}, it can be proved that the two problems \eqref{sec4_prob:AST_h} and \eqref{sec4_prob:atomic_mix_int} are equivalent. Readers can refer to \cite{Coordinate_ANM_Li_2024} for the detailed proof. We then implement the framework proposed in \cite{Coordinate_ANM_Li_2024} for updating the channel vector $\mathbf{h}$ by inexactly solving \eqref{sec4_prob:atomic_mix_int}. We summarize the framework in Algorithm \ref{sec4_alg:AST} for completeness.

\begin{algorithm}[t]
        \caption{Algorithm for inexactly solving \eqref{sec4_prob:atomic_mix_int}}
        \begin{algorithmic}[1]\label{sec4_alg:AST}
                \REQUIRE Gradient step size $\alpha$; Result of gradient step $\hat{\mathbf{h}}$; Atomic set $\mathcal{A}$; Threshold $\eta$; Tolerance $\varepsilon$; Maximum iteration $K_{\max}$.
                \STATE Initialize with empty set of tuples $\mathcal{\mathcal{U}}$; Residual vector $\mathbf{h}_{r} \gets \hat{\mathbf{h}}$; $L \gets 0$; $i \gets 1$; $\delta \gets \varepsilon/ (\alpha \|\hat{\mathbf{h}}\|_{2}^{2}/\eta + \varepsilon)$; $\eta^{\prime} \gets (1-\delta) \eta$;
                \STATE $k \gets 0$;
                \REPEAT
                \IF{$i \leq L$}
                \STATE $(c_{i}, \mathbf{a}_{i}) \gets [\mathcal{U}]_{i}$;
                \STATE $\mathbf{h}_{r}^{i} \gets \mathbf{h}_{r} + c_{i} \mathbf{a}_{i}$;
                \STATE $(c_{i}, \mathbf{a}_{i}) \gets \displaystyle{\argmin_{c \in \mathbb{C}, \mathbf{a} \in \mathcal{A}}}\;\frac{\alpha}{2} \|\mathbf{h}_{r}^{i} - c \mathbf{a}\|_{2}^{2} + \eta^{\prime} |c|$;\label{alg_AST_sub_prob1}
                \STATE $\mathbf{h}_{r} \gets \mathbf{h}_{r}^{i} - c_{i} \mathbf{a}_{i}$;
                \STATE $[\mathcal{U}]_{i} \gets (c_{i}, \mathbf{a}_{i})$;
                \IF{$c_{i} == 0$}
                \STATE Remove $(c_{i}, \mathbf{a}_{i})$ from $\mathcal{U}$;
                \STATE $L \gets L-1$, $i \gets i-1$;
                \ENDIF
                \STATE $i \gets i+1$;
                \ELSIF{$\left|\eta \sum_{j=1}^{L} |c_{j}| - \langle \mathbf{h}_{r}, \hat{\mathbf{h}} - \mathbf{h}_{r} \rangle\right| < \varepsilon$}
                \IF{$\sup_{\mathbf{a} \in \mathcal{A}} \langle \mathbf{h}_{r}, \mathbf{a} \rangle > \eta$}
                \STATE $(c_{L+1}, \mathbf{a}_{L+1}) \gets \displaystyle{\argmin_{c \in \mathbb{C}, \mathbf{a} \in \mathcal{A}}}\; \frac{\alpha}{2} \|\mathbf{h}_{r} - c \mathbf{a}\|_{2}^{2} + \eta^{\prime} |c|$;\label{alg_AST_sub_prob2}
                \STATE $\mathbf{h}_{r} \gets \mathbf{h}_{r} - c_{L+1} \mathbf{a}_{L+1}$;
                \STATE Add $(c_{L+1}, \mathbf{a}_{L+1})$ into $\mathcal{U}$;
                \STATE $L \gets |\mathcal{U}|$, $i \gets 1$;
                \ELSE
                \STATE Break;
                \ENDIF
                \ELSE
                \STATE $i\gets 1$;
                \ENDIF
                \STATE $k \gets k+1$;
                \UNTIL{Meet the stop criterion or $k >= K_{\max}$;}
                \ENSURE $\mathcal{U}$.
        \end{algorithmic}
\end{algorithm}

As discussed in \cite{Bhaskar2013} and the related works performing atomic norm denoising, the regularization parameter on the atomic norm directly balances the error and the sparsity of the solution. Following the derivation procedure in \cite[Proposition 1]{Bhaskar2013}, we set the atomic norm regularization parameter as
\begin{equation}\label{eq:threshold}
        \eta = \mathbb{E}_{\mathbf{n}} \sup_{\mathbf{a} \in \mathcal{A}} \langle \mathbf{n}, \mathbf{S} \mathbf{a} \rangle = \sigma_{n} \|\mathbf{s}\|_{2} \left(\sqrt{\frac{\pi}{2}} + 6 \sqrt{1-\frac{\pi}{4}}\right).
\end{equation}
In Algorithm \ref{sec4_alg:AST}, it is crucial to reliably solve the problems in step \ref{alg_AST_sub_prob1} and step \ref{alg_AST_sub_prob2}. Both problems can be treated as conic projections, but due to the nonconvex atomic set $\mathcal{A}$, these problems are nonconvex. Nevertheless, this projection operation fortunately maintains a separable structure \cite{Coordinate_ANM_Li_2024}. Considering the optimization problem in step \ref{alg_AST_sub_prob2} and denoting $\tilde{\eta} = \eta^{\prime} / \alpha$, we have
\begin{equation}
        \begin{aligned}
                 & \frac{1}{2} \|\mathbf{h}_{r} - c \mathbf{a}\|_{2}^{2} + \tilde{\eta} |c| = \frac{1}{2}\left[\|\mathbf{h}_{r}\|_{2}^{2}\right.                                                                                                                                                \\
                 & \left.+ \left(c \|\mathbf{{a}}\|_{2} + \frac{\tilde{\eta} - \left\langle \mathbf{h}_{r}, \mathbf{a} \right\rangle}{\|\mathbf{a}\|_{2}}\right)^2 - \frac{\left(\tilde{\eta} - \left\langle \mathbf{h}_{r}, \mathbf{a} \right\rangle\right)^2}{\|\mathbf{a}\|_{2}^{2}}\right].
        \end{aligned}
\end{equation}
Since the term $\|\mathbf{a}\|_{2}^{2} = (MN)^2$ is irrelevant to the choice of the atom vector $\mathbf{a}$. The solution to the problem is given as
\begin{subequations}
        \begin{align}
                \mathbf{a}^{\star} & = \argmax_{\mathbf{a} \in \mathcal{A}} \left\langle \mathbf{h}_{r}, \mathbf{a} \right\rangle,                     \label{sec4_eq:update_a}                                                                                                                                                                                                                                              \\
                c^{\star}          & = \begin{cases}
                                               0,                                                                                                                                                                  & \left\langle\mathbf{h}_{r}, \mathbf{a}^{*}\right\rangle \leq \tilde{\eta}, \\
                                               \frac{(\mathbf{h}_{r}^{\operatorname{H}} \mathbf{a}^{\star})^{*}}{(MN)^2} \left(1 - \frac{\tilde{\eta}}{|\mathbf{h}^{\operatorname{H}} \mathbf{a}^{\star}|}\right), & \left\langle\mathbf{h}_{r}, \mathbf{a}^{*}\right\rangle > \tilde{\eta}.
                                       \end{cases}
        \end{align}
\end{subequations}
Determining $\mathbf{a}^{\star}$ is equivalent to maximizing $|\mathbf{h}_{r}^{\operatorname{H}}\mathbf{a}|^2$ with respect to $\tau$ and $\nu$. However, due to the non-convexity of the objective function and the presence of multiple local extrema, finding the global maximum is non-trivial. Following the strategy in \cite{Mamandipoor2016,Coordinate_ANM_Li_2024}, we adopt a two-stage approach: first, we initialize the parameters by selecting the candidate from an oversampled grid that yields the largest objective value; Subsequently, we refine it using the Newton-based gradient ascent method to solve \eqref{sec4_eq:update_a}.

\subsubsection{Updating the Information Symbols Vector}
Once the channel vector $\mathbf{h}^{(t+1)}$ is updated via Algorithm \ref{sec4_alg:AST}, the subproblem for updating the information symbols $\mathbf{x}_{\operatorname{dd}, d}$ in \eqref{sec4_prob:penalty_model_anm} becomes
\begin{equation}\label{sec4_prob:update_channel_vector_original}
        \min_{\mathbf{x} \in \bar{\mathcal{S}}} \phi_{\rho, t}(\mathbf{x}),
\end{equation}
where $\mathbf{x} = \mathbf{x}_{\operatorname{dd}, d}$, $\phi_{\rho, t}(\mathbf{x}) = \frac{1}{2} \|\mathbf{y}_{\operatorname{dd}, d}^{(t+1)} - \tilde{\mathbf{H}}^{(t+1)} \mathbf{x}\|_{2}^{2} - \rho \|\mathbf{x}\|_{2}^{2}$, $\mathbf{y}_{\operatorname{dd}, d}^{(t+1)} = \mathbf{y}_{\operatorname{dd}} - \mathbf{H}_{\operatorname{eff}}^{(t+1)} \mathbf{x}_{\operatorname{dd}, p}$ represents the DD domain received symbols vector after removing the estimated pilot response with $\mathbf{x}_{\operatorname{dd}, p}$ denoting the DD domain pilot vector, $\tilde{\mathbf{H}}^{(t+1)} = \mathbf{H}^{(t+1)}_{\operatorname{eff}, [:, \operatorname{\;data\;indces}]}$ is the submatrix by picking the columns of $\mathbf{H}_{\operatorname{eff}}^{(t+1)}$ that corresponds to the position of data symbols in $\mathbf{x}_{\operatorname{dd}}$, $\mathbf{H}_{\operatorname{eff}}^{(t+1)}$ is the DD domain effective channel matrix constructed by the updated channel vector $\mathbf{h}^{(t+1)}$.

We adopt the majorization-minimization (MM) method to solve \eqref{sec4_prob:update_channel_vector_original}. The update through MM method is to solve the following problem:
\begin{equation}\label{sec4_prob:MM_update_x}
        \min_{\mathbf{x} \in \bar{\mathcal{S}}} \psi_{\rho, t}(\mathbf{x} | \mathbf{x}^{(t)}),
\end{equation}
where $\psi_{\rho, t}(\mathbf{x} | \mathbf{x}^{(t)})$ is the majorant of $\phi_{\rho, t}(\mathbf{x})$ at $\mathbf{x}^{(t)}$. Specifically, we consider the majorant as
\begin{multline}
        \psi_{\rho, t}(\mathbf{x} | \mathbf{x}^{(t)}) = \frac{1}{2} \|\mathbf{y}_{\operatorname{dd}, d}^{(t+1)} - \tilde{\mathbf{H}}^{(t+1)} \mathbf{x}\|_{2}^{2} \\
        - 2 \rho \langle \mathbf{x}^{(t)}, \mathbf{x} - \mathbf{x}^{(t)} \rangle - \rho \|\mathbf{x}^{(t)}\|_{2}^{2},
\end{multline}
It is easy to verify that $\psi_{\rho, t}(\mathbf{x} | \mathbf{x}^{(t)})$ is a tight upper bound of $\phi_{\rho, t}(\mathbf{x})$ since $\phi_{\rho, t}(\mathbf{x}^{(t)}) = \psi_{\rho, t}(\mathbf{x}^{(t)} | \mathbf{x}^{(t)})$ and $\psi_{\rho, t}(\mathbf{x} | \mathbf{x}^{(t)}) \geq \phi_{\rho, t}(\mathbf{x})$ for any $\mathbf{x} \in \bar{\mathcal{S}}$.

Rather than solving \eqref{sec4_prob:update_channel_vector_original} exactly to update $\mathbf{x}$, we perform one-step accelerated projected-gradient to improve the computation efficiency, i.e., we update $\mathbf{x}^{(t+1)}$ by
\begin{equation}\label{sec4_eq:update_x}
        \mathbf{x}^{(t+1)} = \Pi_{\bar{\mathcal{S}}} \left(\tilde{\mathbf{x}}^{(t)} - \frac{1}{\beta^{(t)}} \nabla \psi_\rho(\tilde{\mathbf{x}}^{(t)} | \mathbf{x}^{(t)})\right),
\end{equation}
where $\beta^{(t)} = \lambda_{\max}([\tilde{\mathbf{H}}^{(t+1)}]^{\operatorname{H}} \tilde{\mathbf{H}}^{(t+1)})$,
\begin{equation}\label{sec4_eq:update_tilde_x}
        \tilde{\mathbf{x}}^{(t)} = \mathbf{x}^{(t)} + \iota^{(t)} (\mathbf{x}^{(t)} - \mathbf{x}^{(t-1)}),
\end{equation}
with
\begin{equation}
        \iota^{(t)} = \frac{\zeta^{(t-1)} - 1}{\zeta^{(t)}}, \zeta^{(t)} = \frac{1 + \sqrt{1 + 4 [\zeta^{(t-1)}]^2}}{2},
\end{equation}
and with $\zeta^{(-1)} = 0$ and $\mathbf{x}^{(-1)} = \mathbf{x}^{(0)}$. The projection operation in \eqref{sec4_eq:update_x} has a closed-form expression, see \cite{Shao2019} for details.

Note that obtaining the exact value of $\rho_0$ in Lemma \ref{sec4_lemma:penalty_parameter} is difficult in practice. Instead, we use the homotopy technique \cite{bertsekas2016nonlinear}, which starts with a small penalty parameter and increases it gradually, tracking the solution path of the corresponding penalty problems \cite{wu2024quantized}. This procedure eventually identifies a $\rho_0$ that satisfies Lemma \ref{sec4_lemma:penalty_parameter}. This technique also typically yields much better numerical performance than solving the penalty model \eqref{sec4_prob:penalty_model_anm} directly with an excessively large $\rho$ \cite{Shao2019, wu2024quantized}. The upper bound $\rho_{\rm upb}$ is selected according to the analysis in [22]. In our problem, for an $M_{\rm PSK}$-PSK constellation, we set $\rho_{\rm upb} = \frac{\bar L}{\sin(\pi/M_{\rm PSK})}$ and $\bar L=\max_t L_t$, with $L_t = \|\tilde{\mathbf H}^{(t+1)}\|_2^2\sqrt{N_d} + \|\tilde{\mathbf H}^{(t+1)}\|_2 \|\mathbf{y}_{\rm dd,d}^{(t+1)}\|_2$.

Although the LMMSE or other detectors can be applied to update the symbol detection, it does not explicitly exploit the discrete PSK constellation. In contrast, the negative square penalty incorporates the PSK constraint into a unified optimization framework, promotes solutions toward constellation points, and enables a simple symbol update with tractable convergence analysis.

\begin{algorithm}[t]
        \caption{Algorithm for Bistatic Zak-OTFS ISAC}
        \begin{algorithmic}[1]\label{alg:Bistatic_ISAC}
                \REQUIRE DD domain receive signal $\mathbf{y}_{\operatorname{dd}}$; Atomic set $\mathcal{A}$; Threshold $\eta$; Summable tolerance sequence $\{\varepsilon^{(t)}\}_{t \geq 0}$; Positive integer $n$; Penalty parameter upper bound $\rho_{\operatorname{upb}}$; Maximum iterations $T_{\max}$; Positive constant $\delta_x, \delta_h$.
                \STATE Initialize according to Section \ref{sec4:initialization};
                \STATE $t\gets 0$;
                \REPEAT
                \STATE Inexactly solve \eqref{sec4_prob:penalty_model_anm} via Algorithm \ref{sec4_alg:AST} with tolerance $\varepsilon^{(t)}$, and obtain $\mathbf{h}^{(t+1)}$ with $\mathcal{U}^{(t+1)}$;
                \STATE Reconstruct the DD domain effective channel with $\mathcal{U}^{(t+1)}$ through \eqref{sec2_eq:eff_DD_channel_mtx}, and obtain $\mathbf{H}_{\operatorname{eff}}^{(t+1)}$;
                \STATE Update information symbols sequence $\mathbf{x}_{\operatorname{dd}, d}^{(t+1)}$ via \eqref{sec4_eq:update_x};
                \IF{$\|\mathbf{x}_{\operatorname{dd}, d}^{(t+1)} - \mathbf{x}_{\operatorname{dd}, d}^{(t)}\|_{2}^{2} \leq \delta_x$ or $[t]_{n} = 0$}
                \STATE $\rho^{(t+1)} \gets \min\{c \rho^{(t)}, \rho_{\operatorname{upb}}$\};\label{alg2_step:penalty_update}
                \ENDIF
                \STATE Obtain $\mathbf{s}^{(t+1)}$ via \eqref{sec2_eq:Zak_OTFS_modulation};
                \STATE $t\gets t+1$;
                \UNTIL{$\|\mathbf{h}^{(t+1)} - \mathbf{h}^{(t)}\|_{2} \leq \delta_h$ and $\|\mathbf{x}_{\operatorname{dd}, d}^{(t+1)} - \mathbf{x}_{\operatorname{dd}, d}^{(t)}\|_{2}^{2} \leq \delta_x$ or $t$ reach maximum iterations $T_{\max}$;}
                \ENSURE $(\mathbf{h}^{(t)}, \mathbf{x}_{\operatorname{dd}, d}^{(t)})$.
        \end{algorithmic}
\end{algorithm}

\subsection{Computational Complexity} 

We briefly analyze the computational complexity of the proposed algorithm. Let $Q=MN$ denote the number of DD-domain symbols in one Zak-OTFS frame, and let $N_d$ denote the number of data symbols. The main computational burden comes from the channel update step, where Algorithm \ref{sec4_alg:AST} is used to solve the atomic norm denoising subproblem. In the proposed formulation, each atom is given by $\mathbf a(\tau,\nu)=\mathbf b_{\tau}\otimes \mathbf d_{\nu}$, where $\mathbf b_{\tau},\mathbf d_{\nu}\in\mathbb C^{Q}$. Hence, the atom dimension is $Q^2$. For the considered two-dimensional DD atoms, this projection is implemented by an oversampled two-dimensional FFT search followed by local Newton refinement. Therefore, the per-step complexity is $\mathcal O\left(Q^2\log Q+K_{\rm N}Q^2\right)$, where $K_{\rm N}$ is the number of Newton refinement iterations. Treating $K_{\rm N}$ as a small constant, the per-step complexity becomes $\mathcal O(Q^2\log Q)$. If the final solution contains $L$ active atoms, the coordinate descent solver requires approximately $\mathcal O(L^2)$ iterations to terminate \cite{Coordinate_ANM_Li_2024}. Thus, the complexity of Algorithm \ref{sec4_alg:AST} for one channel update is $\mathcal O\left(L^2Q^2\log Q\right)$. For $T_{\rm out}$ outer iterations, the dominant complexity of the proposed receiver is therefore $\mathcal O\left(T_{\rm out}L^2Q^2\log Q\right).$ The data-symbol update is much cheaper. Given the reconstructed effective channel matrix, it mainly requires matrix-vector multiplications with $\tilde{\mathbf H}\in\mathbb C^{Q\times N_d}$, whose complexity is $\mathcal O(QN_d)$ per update, while the projection onto the convex hull of the PSK constellation is element-wise.

The practical feasibility of the proposed algorithm mainly depends on the sparsity of the DD domain channel. As discussed in \cite{Coordinate_ANM_Li_2024}, the coordinate-descent ANM solver is efficient for sparse AST problems. Therefore, the proposed method is suitable for sparse high-mobility ISAC scenarios with only a few dominant targets or paths. In such cases, \(L\ll Q^2\), and the coordinate-descent solver provides a more scalable alternative to SDP-based ANM while enabling off-grid delay and Doppler estimation. However, when the channel becomes dense or the frame size is very large, the \(L^2\) factor may lead to high computational cost. Hence, the proposed receiver is mainly intended for sparse ISAC applications requiring accurate parameter estimation, rather than for general low-complexity detection.

\subsection{Convergence Analysis}
The procedures for performing Zak-OTFS ISAC is summarized in Algorithm \ref{alg:Bistatic_ISAC}. We now analyze its convergence property in this subsection. For the ease of representation, we simplify the notations here. We rewrite the problem \eqref{sec4_prob:penalty_model_anm} as $\min_{\mathbf{x}, \mathbf{h}} F(\mathbf{x}, \mathbf{h})$, where $F(\mathbf{x}, \mathbf{h})$ is given by
\begin{equation}
        F(\mathbf{x}, \mathbf{h}) = f(\mathbf{x}, \mathbf{h}) + \mathbb{I}_{\bar{\mathcal{S}}}(\mathbf{x}),
\end{equation}
$\mathbf{x} = \mathbf{x}_{\operatorname{dd}, d}$, and $f(\mathbf{x}, \mathbf{h})$ is the objective function in \eqref{sec4_prob:penalty_model_anm}. We claim the following facts to help the convergence analysis.
\begin{enumerate}
        \item $\alpha^{(t)} = \|\mathbf{s}^{(t)}\|_{2}^2$ is bounded by
              \begin{equation}\label{eq:alpha_bound}
                      0 < \munderbar{C}_{s} \leq \alpha^{(t)} \leq \bar{C}_{s} < \infty, \forall t.
              \end{equation}
        \item $\beta^{(t)} = \lambda_{\max}([\tilde{\mathbf{H}}^{(t+1)}]^{\operatorname{H}} \tilde{\mathbf{H}}^{(t+1)})$ is bounded by
              \begin{equation}\label{eq:beta_bound}
                      0 < \munderbar{C}_{h} \leq \beta^{(t)} \leq \bar{C}_{h} < \infty, \forall t.
              \end{equation}
        \item $\{\varepsilon^{(t)}\}_{t \geq 0}$ is summable and non-increasing.
\end{enumerate}
The boundness conditions in \eqref{eq:alpha_bound} and \eqref{eq:beta_bound} follow from the finite and nonzero energies of the transmitted signal and reconstructed channel matrix. With these notations, we now present the convergence result of Algorithm \ref{alg:Bistatic_ISAC} in the following theorem.
\begin{theorem}\label{thm:convergence_result}
        Suppose that $\mu^{(t)}$, $\iota^{(t)}$ in Algorithm \ref{alg:Bistatic_ISAC} satisfy
        \begin{equation}
                0 \leq \mu^{(t)} \leq \bar{\mu}, 0 \leq \iota^{(t)} \leq \bar{\iota}, \forall t,
        \end{equation}
        where $\bar{\mu} = \sqrt{\frac{\munderbar{C}_{s}}{\bar{C}_{s}} (1 - \theta)}$ and $\bar{\iota} = \sqrt{\frac{\munderbar{C}_{h}}{\bar{C}_{h}} (1 - \theta)}$ with $0 < \theta \leq 1$. Then, Algorithm \ref{alg:Bistatic_ISAC} will achieve $\delta$ away to the $\varepsilon$-stationary point within $\mathcal{O}(1 / \delta^2)$ iterations, i.e.,
        \begin{equation}
                \min_{t^{\prime} = 0, 1, \dots, t} \operatorname{dist}\left(\mathbf{0}, \partial_{\varepsilon} F(\mathbf{x}^{(t^{\prime}+1)}, \mathbf{h}^{(t^{\prime}+1)})\right) \leq \frac{C}{\sqrt{t}},
        \end{equation}
        where the explicit expression of the finite number $C$ is given in Appendix \ref{proof:sufficient_descent}.
\end{theorem}
\begin{proof}
        See Appendix \ref{proof:sufficient_descent}.
\end{proof}

Although we only proved the algorithm attains an $\varepsilon$-stationary point at a rate of $\mathcal{O}(1 / \sqrt{t})$, the derivation provided in Appendix \ref{proof:sufficient_descent} directly implies a corollary: any limit point of the generated sequence is a stationary point. We omit the proof of this corollary for brevity.

\section{Numerical Results}\label{sec5:numerical_results}

In this section, we evaluate the convergence, communication, and sensing performance of the proposed algorithm. We adopt the Zak-OTFS frame structure in Fig. \ref{sec2_fig:embedded_pilot_illustration} with $M=8$ and $N=16$. The carrier frequency is $24$ GHz, and the subcarrier spacing is $\Delta f=30$ kHz. The pilot is placed at grid $(4,8)$, with the pilot region and guard space given by $[2,6]\times[5,11]$ and $[1,7]\times[4,12]$, respectively. The data symbols have unit magnitude, while the pilot magnitude is set to $\sqrt{63}$ such that the total transmit power is $MN$. We consider a bistatic ISAC scenario with several targets and no line-of-sight (LoS) link between the base station and the receiver. Without loss of generality, the delay of the first target is fixed to be zero, while the delays of the remaining targets are drawn uniformly from $(0,1.5]/(M\Delta f)$ subject to a minimum pairwise separation of $\kappa/(MN\Delta f)$. The Doppler shifts are drawn uniformly from $[-1.5,1.5]/(NT)$ subject to a minimum pairwise separation of $\kappa/(MNT)$. The channel gains follow $\mathcal{CN}(0,1)$\footnote{Although target RCS fluctuations are commonly modeled using Swerling models, we adopt this simplified setting following prevalent ISAC studies. The proposed framework can be extended to fluctuating targets by incorporating the corresponding amplitude statistics into the threshold in \eqref{eq:threshold}.}. The receiver knows only the frame structure and pilot configuration, while the embedded data symbols remain unknown. All results are averaged over $5000$ Monte Carlo trials with independently generated channel parameters and data symbols.

\subsection{Convergence Performance of Proposed Algorithm}

\begin{figure}[t]
        \centering
        \subfloat[Objective function values.]{\includegraphics[width=0.75\columnwidth]{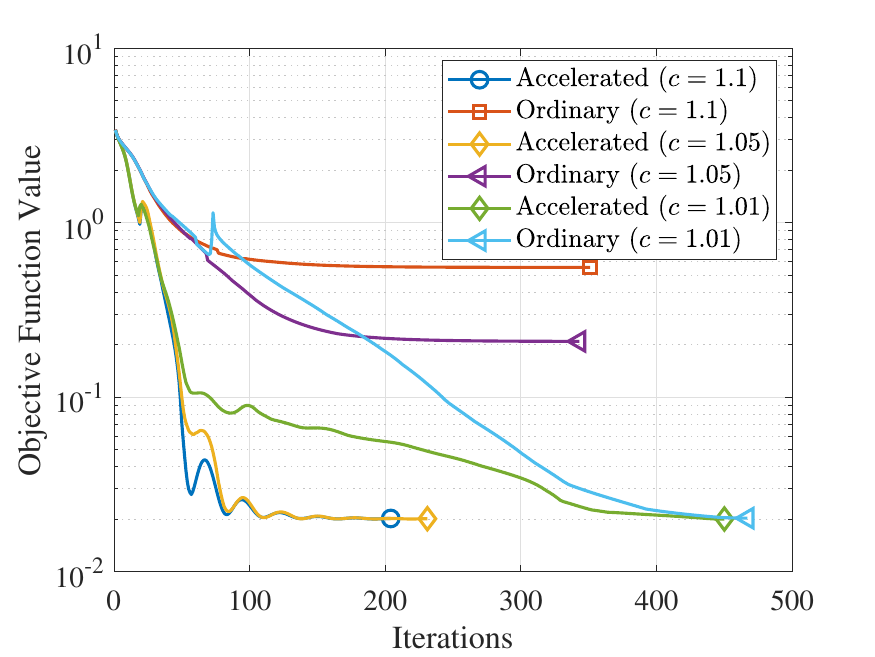}\label{fig:obj_compare}}
        \hfill
        \subfloat[Number of inner iterations.]{\includegraphics[width=0.75\columnwidth]{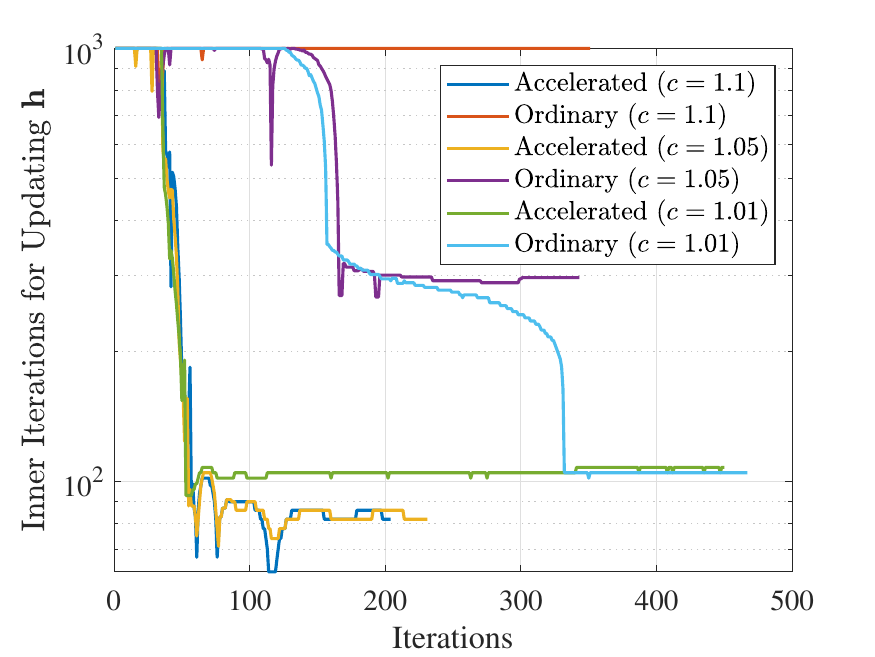}\label{fig:inner_compare}}
        \caption{Iteration performance of the proposed accelerated algorithm and ordinary algorithm.}\label{fig:convergence}
\end{figure}

Fig. \ref{fig:convergence} compares the convergence behavior of the accelerated and ordinary schemes for different penalty-parameter update factors $c$ used in step \ref{alg2_step:penalty_update} of Algorithm \ref{alg:Bistatic_ISAC}. The ordinary scheme is obtained by setting the extrapolation parameters $\mu^{(t)}$ and $\iota^{(t)}$ in \eqref{sec4_eq:mid_h_acc_proximal} and \eqref{sec4_eq:update_tilde_x}, respectively, to zero for all $t$. As shown in Fig. \ref{fig:obj_compare}, the accelerated schemes achieve a substantially faster decrease in the objective value. When $c=1.1$, the accelerated scheme reaches a low objective value within approximately $120$ outer iterations, whereas the ordinary scheme stagnates at a considerably higher level. A similar performance gap is observed for $c=1.05$.

The degradation of the ordinary scheme under relatively large values of $c$ is mainly caused by the slow update of the channel vector. As the penalty parameter increases rapidly, the symbol vector is pushed toward the vertices of the constellation convex hull before the channel estimate becomes sufficiently accurate. The symbols may prematurely converge to incorrect vertices. These symbol detection errors further distort the channel update, causing the algorithm to be trapped at an inferior stationary point with a relatively large objective value. When a smaller update factor is adopted, e.g., $c=1.01$, the symbol vector approaches the constellation vertices more gradually, allowing sufficient refinement of the channel estimate. Therefore, both schemes can eventually attain comparable objective values, although such a small penalty-parameter update factor slows the convergence of the symbol vector.

The number of inner iterations shown in Fig. \ref{fig:inner_compare} further confirms the computational advantage of the accelerated scheme. After a short transient stage, the accelerated schemes rapidly reduce the required number of inner iterations to around $10^2$ and then remain relatively stable for all tested values of $c$. In contrast, the ordinary scheme requires substantially more inner iterations. In particular, when $c=1.1$, the ordinary scheme almost always reaches the maximum number of inner iterations, i.e., $10^3$ in our simulations, indicating severe difficulty in solving the channel-update subproblem. For $c=1.05$ and $c=1.01$, the required inner iterations of the ordinary scheme eventually decrease, but only after a much longer transient stage. These results show that the accelerated scheme not only improves the outer convergence behavior but also reduces the computational burden of each outer iteration, making it less sensitive to the choice of the penalty-parameter update factor.

\subsection{Communication Performance}

Fig. \ref{fig:ber_bpsk} compares the BER performance of the proposed algorithm under BPSK and QPSK modulations. The benchmarks include an LMMSE detector with perfect CSI and a standard LMMSE detector using the channel estimate obtained from \eqref{sec2_channel_estimation_scheme}. The proposed algorithm consistently approaches the perfect-CSI lower bound. The remaining gap from the lower bound is mainly caused by residual channel-estimation errors and becomes slightly more noticeable in the multi-target case because of interference between the pilot responses. Moreover, the BER is lower for $P=2$ than for $P=1$, indicating that the proposed algorithm effectively exploits the DD domain diversity provided by resolvable propagation paths \cite{Li2021}. Besides, BPSK achieves a lower BER than QPSK at the same SNR, since its constellation points are more widely separated and therefore more robust to noise and channel-estimation errors.

\begin{figure}[t]
        \centering
        \includegraphics[width=0.85\columnwidth]{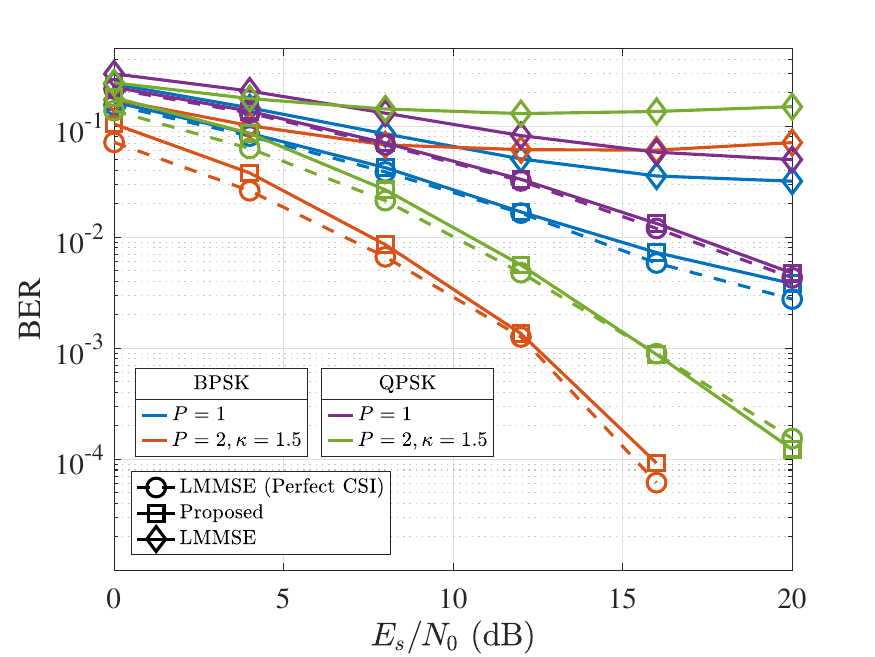}
        \caption{BER performance under different modulations.}\label{fig:ber_bpsk}
\end{figure}

\begin{figure}[t]
        \centering
        \includegraphics[width=0.85\columnwidth]{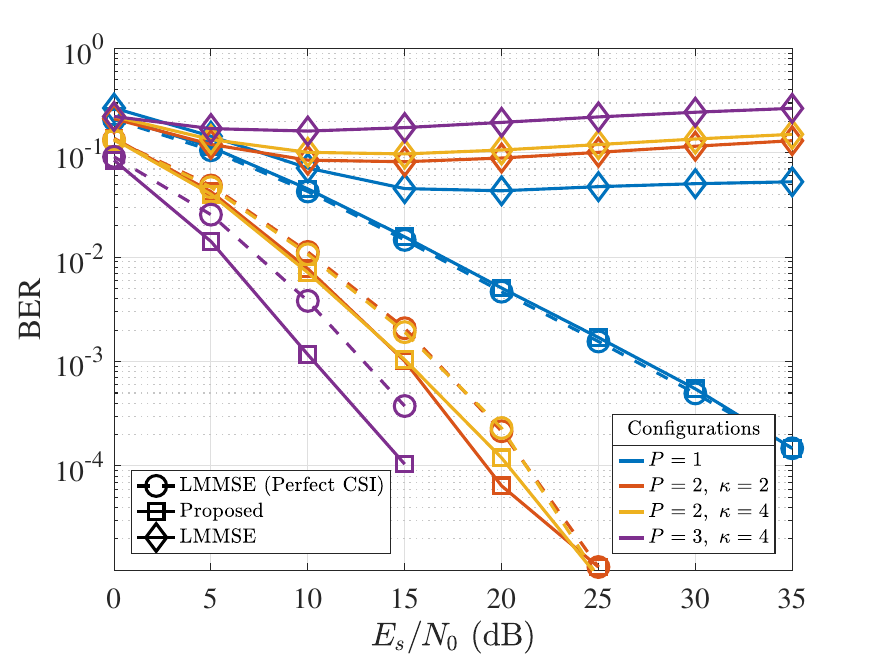}
        \caption{BER performance under QPSK modulation with different settings.}\label{fig:ber_qpsk_settings}
\end{figure}

Fig. \ref{fig:ber_qpsk_settings} compares the BER performance under QPSK modulation for different numbers of targets and minimum target separations. For the standard LMMSE detector, the BER degrades as the number of targets increases. Moreover, for $P=2$, a larger separation parameter, i.e., $\kappa=4$, leads to a slightly higher BER than $\kappa=2$ at high SNR. This is because an increased number of targets or a larger target separation generally broadens the support of the effective channel response in the DD domain. Since the pilot region is of limited size, a larger fraction of the pilot response may fall outside the observation region, resulting in less accurate channel reconstruction and pronounced BER floors for the standard LMMSE detector.

In contrast, the proposed algorithm jointly refines the channel parameters and data symbols and therefore substantially improves the detection performance. The corresponding BERs decrease consistently with the SNR for all tested configurations. In the multi-target cases, the proposed algorithm can even outperform the LMMSE detector with perfect CSI. This does not imply a channel estimate more accurate than perfect CSI; rather, it shows the benefit of negative penalty scheme in symbol detection, which outperforms the linear LMMSE detection. The results again show that the proposed method effectively exploits the DD diversity provided by multiple resolvable propagation paths.

\subsection{Sensing Performance}

The sensing performance is evaluated using the same Monte-Carlo samples as the communication performance evaluation. To evaluate multi-target performance, we perform data association between the ground truths and estimates using the Hungarian algorithm \cite{kuhn1955hungarian} based on the absolute error of the DD parameter pairs. A valid detection is declared only if the matched estimate falls within a predefined distance gate (half of resolution in our simulation). Otherwise, it will be considered as a false alarm. The probability of detection is calculated as the ratio of these valid matches to the total number of ground truth targets, and the false alarm rate is calculated as the ratio of false alarms to the total number of detects.

\begin{figure}[t]
        \centering
        \subfloat[The estimated parameters.]{\includegraphics[width=0.45\columnwidth]{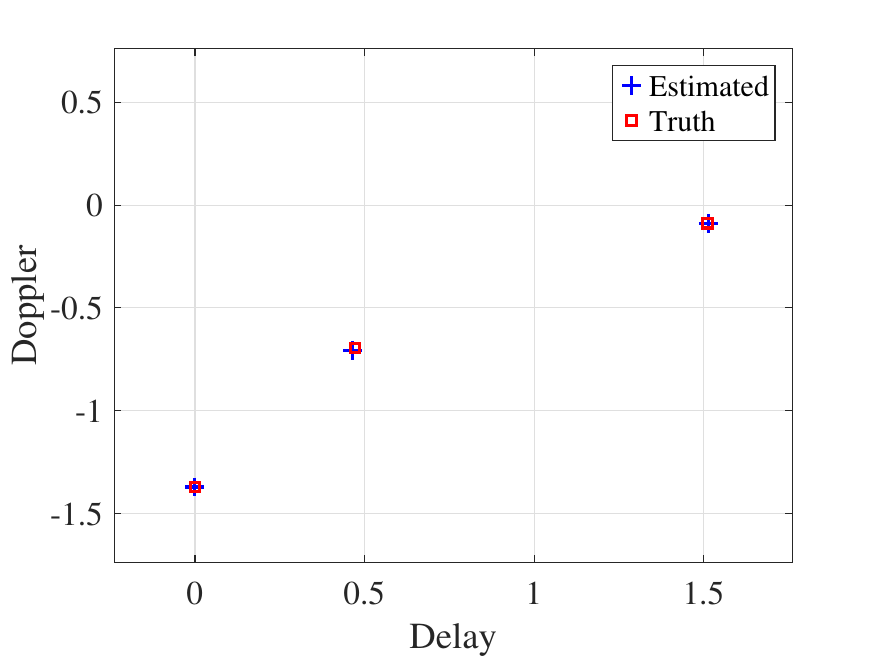}\label{fig:params_example}}
        \subfloat[The DD domain receive signal.]{\includegraphics[width=0.5\columnwidth]{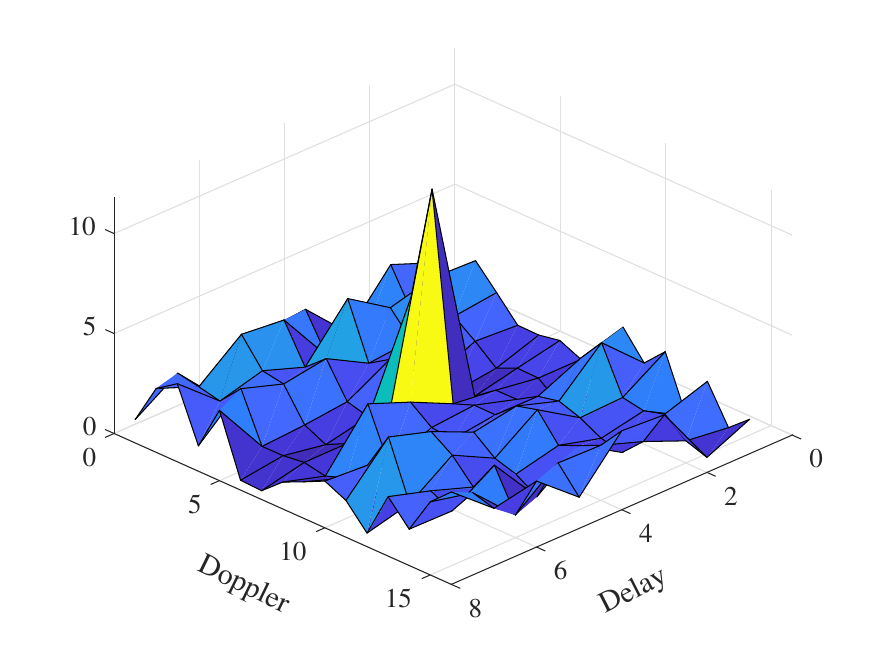}\label{fig:DD_signal}}
        \caption{Illustrations of parameter estimation and the corresponding received DD domain signal.}\label{fig:illustration}
\end{figure}

Fig. \ref{fig:illustration} illustrates a representative example of the proposed parameter-estimation result. Fig. \ref{fig:params_example} compares the estimated delay-Doppler parameters with the ground-truth values, and Fig. \ref{fig:DD_signal} shows the corresponding received DD-domain signal. It can be observed that the proposed method can accurately recover the off-grid delay-Doppler parameters. However, although the proposed ANM-based scheme enables off-grid delay-Doppler estimation, it cannot resolve arbitrarily close targets. Its detection and estimation performance depends on both the number of targets and their minimum delay-Doppler separation; closely spaced targets may be merged or inaccurately estimated due to highly coherent atoms. In the semi-blind setting, symbol-detection errors may further degrade the practical resolution. This limitation is further demonstrated through the simulation results presented later in the paper.

Fig. \ref{fig:detection_qpsk_settings} evaluates the target-detection performance under QPSK modulation. The single-target case achieves an almost perfect detection probability over the entire SNR range. In multi-target scenarios, the detection probability is lower at low SNR because the overlapping DD responses make weak targets more difficult to distinguish. This effect becomes more noticeable when more targets are present or when the target separation is smaller. As the SNR increases, the detection probability approaches one for all configurations, and the false-alarm probability decreases significantly. A larger target separation generally reduces the false-alarm probability by mitigating the mutual interference between target responses. These results demonstrate that the proposed algorithm can reliably detect multiple targets, provided that the received SNR is sufficiently high.

\begin{figure}
        \centering
        \includegraphics[width=0.85\columnwidth]{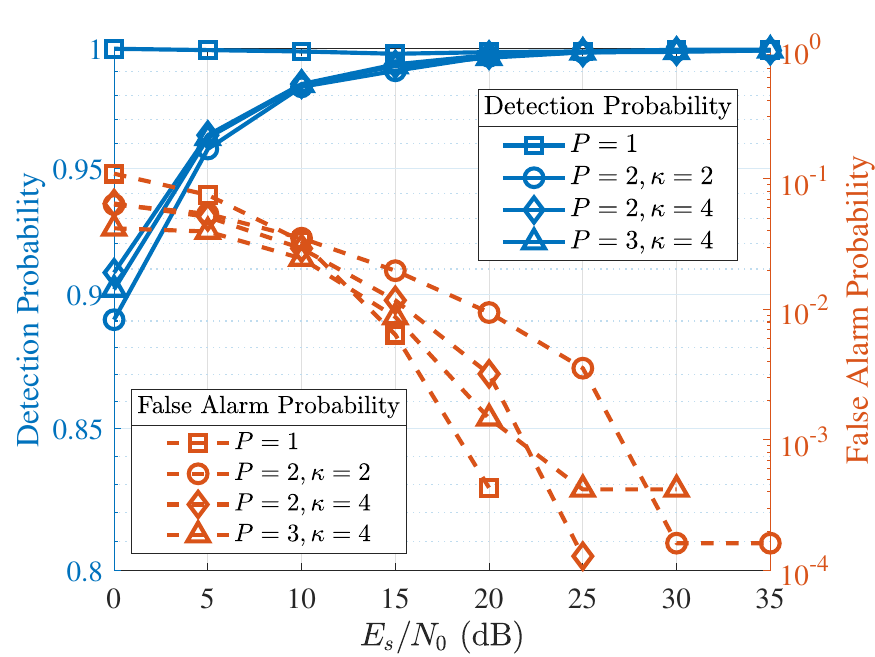}
        \caption{Target detection performance under QPSK modulation with different settings.}\label{fig:detection_qpsk_settings}
\end{figure}

Fig. \ref{fig:sensing_qpsk} evaluates the delay, Doppler, and channel estimation performance under QPSK modulation. We consider a benchmark where the transmitted symbols are perfectly known at the receiver. As expected, the estimation errors decrease with the increasing SNR. The delay and Doppler estimation RMSEs shown in Fig. \ref{fig:sensing_qpsk_range_RMSE} and Fig. \ref{fig:sensing_qpsk_velocity_RMSE} are consistently lower than those obtained in the bistatic ISAC setting, since the unknown data symbols introduce additional uncertainty into the joint estimation problem. This gap is particularly noticeable in the single-target case. As the SNR increases, the two curves gradually approach each other because improved channel estimation leads to more reliable symbol detection, which in turn further refines the channel estimate.

\begin{figure*}
        \centering
        \subfloat[Delay estimation performance.\label{fig:sensing_qpsk_range_RMSE}]{\includegraphics[width=0.33\textwidth]{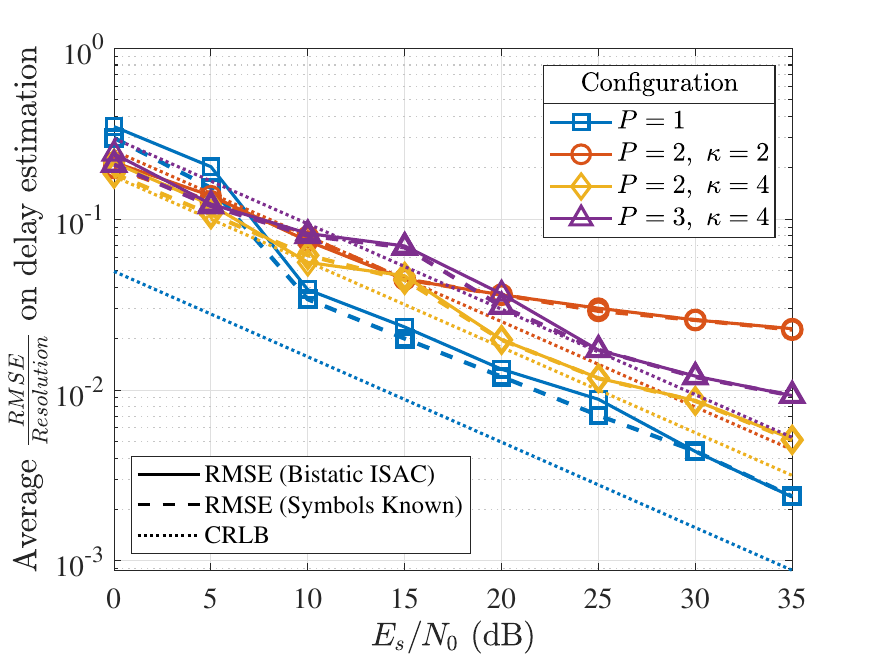}}
        \subfloat[Doppler estimation performance.\label{fig:sensing_qpsk_velocity_RMSE}]{\includegraphics[width=0.33\textwidth]{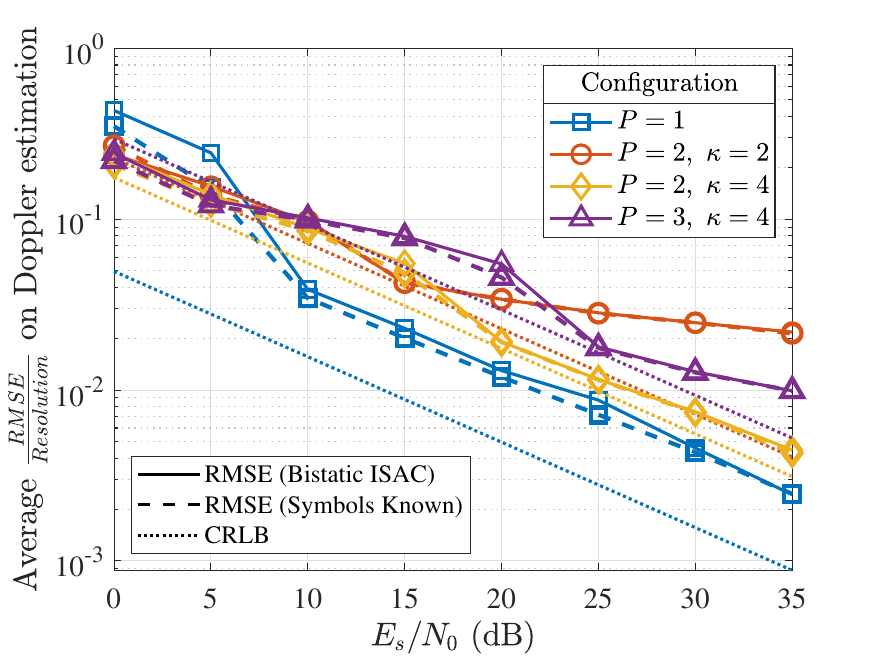}}
        \subfloat[Channel estimation NMSE.\label{fig:sensing_qpsk_ch_est_RMSE}]{\includegraphics[width=0.33\textwidth]{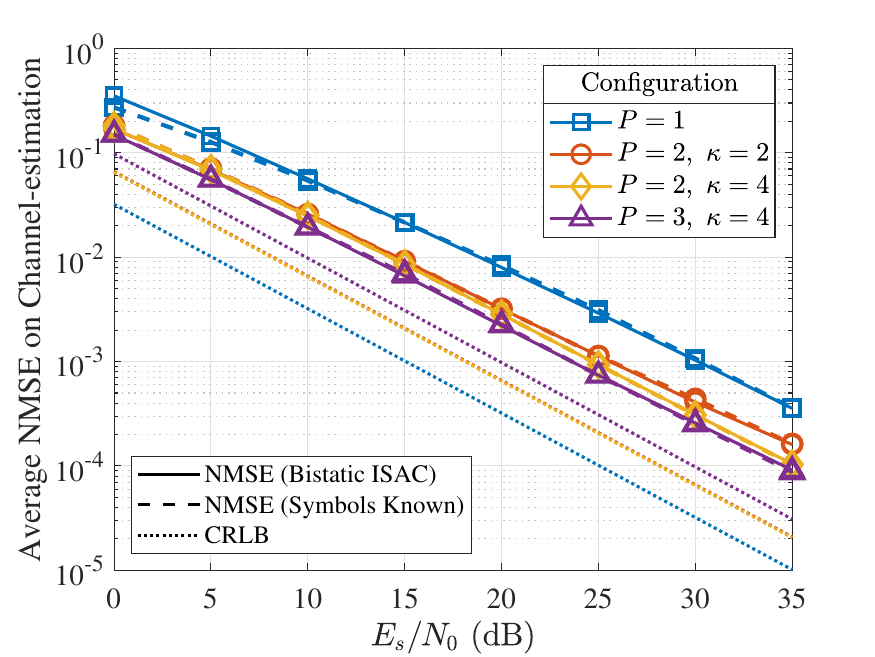}}
        \caption{Sensing performance with QPSK modulation.}\label{fig:sensing_qpsk}
\end{figure*}

In Fig. \ref{fig:sensing_qpsk_ch_est_RMSE}, the channel estimation NMSE of the single-target case exhibits a relatively large gap from the CRLB because residual gain-estimation errors account for a larger fraction of the total channel energy when only one path is present. At high SNR, increasing the number of resolvable targets or enlarging their separation generally improves the NMSE. This trend differs from that of the parameter-estimation performance in Fig. \ref{fig:sensing_qpsk_range_RMSE} and Fig. \ref{fig:sensing_qpsk_velocity_RMSE} because the channel estimation NMSE measures the reconstruction accuracy of the aggregate channel response rather than the accuracy of each individual path parameter. More widely separated DD responses are easier to distinguish, while multiple resolvable paths provide additional diversity and reduce the relative impact of the estimation error after normalization by the total channel energy.

\section{Conclusion}\label{sec6:conclusion}

In this paper, we developed a semi-blind receiver for bistatic Zak-OTFS ISAC, where only the pilot configuration and frame structure are known at the receiver. We first derived a discrete-time matrix I/O relation and showed that the predictability property of Zak-OTFS is preserved, enabling model-free channel estimation for receiver initialization. The joint off-grid channel parameter estimation and PSK symbol detection problem was then formulated as an atomic norm denoising problem. To efficiently solve the resulting problem, we developed an accelerated iterative algorithm combining MM method, inexact accelerated proximal gradient channel updates, and accelerated projected gradient symbol updates. We further established convergence to an $\varepsilon$-stationary point. Simulation results demonstrated accelerated and robust convergence, super-resolution delay-Doppler estimation, reliable multi-target detection, and superior BER performance.

{
\appendices
\section{Proof of Theorem \ref{thm:convergence_result}}\label{proof:sufficient_descent}
The proof proceeds in two steps. First, using the properties of the inexact accelerated proximal-gradient update for the channel vector and the accelerated projected-gradient update for the symbol vector, we derive a sufficient-descent inequality for each iteration. Second, by summing these inequalities and invoking the lower-boundedness of the objective function, we bound the accumulated stationarity residuals and establish sublinear convergence to an $\varepsilon$-stationary point. The following lemma provides a key ingredient for this analysis.

\begin{lemma}[\cite{Schmidt2011}]\label{lm:epsilon_solution_property}
        If $\mathbf{h}^{(t+1)}$ is an $\varepsilon^{(t)}$-optimal solution to the proximal problem \eqref{sec4_prob:AST_h_proximal}, then there exists a vector $\bm{\delta}^{(t)}$ such that $\|\bm{\delta}^{(t)}\|_{2} \leq \sqrt{\frac{2 \varepsilon^{(t)}}{\alpha^{(t)}}}$ and
        \begin{multline}
                \alpha^{(t)} [ \tilde{\mathbf{h}}^{(t)} - \mathbf{h}^{(t+1)} - \frac{1}{\alpha^{(t)}} \nabla  \gamma_{t}(\tilde{\mathbf{h}}^{(t)}) - \bm{\delta}^{(t)} ]\\
                \in \partial_{\varepsilon^{(t)}} \varphi(\mathbf{h}^{(t+1)}),
        \end{multline}
        where $\partial_{\varepsilon^{(t)}} \varphi(\mathbf{h}^{(t+1)})$ is the $\varepsilon^{(t)}$-subgradient of $\varphi(\mathbf{h})$ at $\mathbf{h}^{(t+1)}$ \cite[Section 4.3]{bertsekas2003convex} and $\tilde{\mathbf{h}}^{(t)}$ is defined in \eqref{sec4_eq:mid_h_acc_proximal}.
\end{lemma}

\subsection{Upper Boundness of Approximate Stationarity}

The update rule of \eqref{sec4_eq:update_x} is equivalent to
\begin{equation}
        \mathbf{x}^{(t+1)} = \argmin_{\mathbf{x}} \frac{\beta^{(t)}}{2} \|\mathbf{x} -  \hat{\mathbf{x}}^{(t)}\|_{2}^{2} + \mathbb{I}_{\bar{\mathcal{S}}}(\mathbf{x}),
\end{equation}
where $\hat{\mathbf{x}}^{(t)} = \tilde{\mathbf{x}}^{(t)}  - \frac{1}{\beta^{(t)}} \nabla \psi(\tilde{\mathbf{x}}^{(t)} | \mathbf{x}^{(t)})$. According to the first-order optimality, it gives
\begin{equation}
        \mathbf{0} \in \beta^{(t)} (\mathbf{x}^{(t+1)} - \tilde{\mathbf{x}}^{(t)}) + \nabla \psi(\tilde{\mathbf{x}}^{(t)} | \mathbf{x}^{(t)}) + \partial \mathbb{I}_{\bar{\mathcal{S}}}(\mathbf{x}^{(t+1)}).
\end{equation}
Moreover, since $\mathbf{h}^{(t+1)}$ is an $\varepsilon^{(t)}$-optimal solution to the proximal problem \eqref{sec4_prob:AST_h_proximal}, according to Lemma \ref{lm:epsilon_solution_property}, we have
\begin{equation}
        \mathbf{0} \in \alpha^{(t)} (\mathbf{h}^{(t+1)} - \tilde{\mathbf{h}}^{(t)} + \bm{\delta}^{(t)}) + \nabla \gamma_{t}(\tilde{\mathbf{h}}^{(t)}) + \partial_{\varepsilon^{(t)}} \psi(\mathbf{h}^{(t+1)}).
\end{equation}

By denoting $\mathbf{e}_{\mathbf{x}}^{(t+1)} = \nabla \phi_{\rho, t}(\mathbf{x}^{(t+1)}) + \mathbf{v}_{\mathbf{x}}^{(t+1)}$ and $\mathbf{e}_{\mathbf{h}}^{(t+1)} = \nabla \gamma_{t+1}(\mathbf{h}^{(t+1)}) + \mathbf{v}_{\mathbf{h}, \varepsilon^{(t)}}^{(t+1)}$ with $\mathbf{v}_{\mathbf{x}}^{(t+1)} \in \partial \mathbb{I}_{\bar{\mathcal{S}}}(\mathbf{x}^{(t+1)})$ and $\mathbf{v}_{\mathbf{h}, \varepsilon^{(t)}}^{(t+1)} \in \partial_{\varepsilon^{(t)}} \psi(\mathbf{h}^{(t+1)})$ being such that
\begin{equation}
        \begin{aligned}
                \mathbf{0} & = \beta^{(t)} (\mathbf{x}^{(t+1)} - \tilde{\mathbf{x}}^{(t)}) + \nabla \psi(\tilde{\mathbf{x}}^{(t)} | \mathbf{x}^{(t)}) + \mathbf{v}_{\mathbf{x}}^{(t+1)},                            \\
                \mathbf{0} & = \alpha^{(t)} (\mathbf{h}^{(t+1)} - \tilde{\mathbf{h}}^{(t)} + \bm{\delta}^{(t)}) + \nabla \gamma_{t}(\tilde{\mathbf{h}}^{(t)}) + \mathbf{v}_{\mathbf{h}, \varepsilon^{(t)}}^{(t+1)},
        \end{aligned}
\end{equation}
we then have
\begin{equation}
        \begin{aligned}
                \mathbf{e}_{\mathbf{x}}^{(t+1)} & \in \partial_{\mathbf{x}} F(\mathbf{x}^{(t+1)}, \mathbf{h}^{(t+1)}),                    \\
                \mathbf{e}_{\mathbf{h}}^{(t+1)} & \in \partial_{\mathbf{h}, \varepsilon^{(t)}} F(\mathbf{x}^{(t+1)}, \mathbf{h}^{(t+1)}),
        \end{aligned}
\end{equation}
where $\partial_{\mathbf{x}} F(\mathbf{x}^{(t+1)}, \mathbf{h}^{(t+1)})$ is the subgradient of $F(\mathbf{x}, \mathbf{h}^{(t+1)})$ towards $\mathbf{x}^{(t+1)}$ and $\partial_{\mathbf{h}, \varepsilon^{(t)}} F(\mathbf{x}^{(t+1)}, \mathbf{h}^{(t+1)})$ represents the $\varepsilon^{(t)}$-subgradient of $F(\mathbf{x}^{(t+1)}, \mathbf{h})$ at $\mathbf{h}^{(t+1)}$. Recap the properties about $\varepsilon$-subgradient in \cite{bertsekas2003convex},
\begin{equation}
        \begin{aligned}
                \cap & \partial_{\varepsilon \downarrow 0} \psi(\mathbf{h}) = \partial \psi(\mathbf{h}),                                                                 \\
                     & \partial_{\varepsilon^{(1)}}(\mathbf{h}) \subset \partial_{\varepsilon^{(2)}}(\mathbf{h}), \text{ if } 0 < \varepsilon^{(1)} < \varepsilon^{(2)},
        \end{aligned}\nonumber
\end{equation}
we then have
\begin{equation}
        \partial_{\mathbf{x}} F(\mathbf{x}^{(t+1)}, \mathbf{h}^{(t+1)}) \subset \partial_{\mathbf{x}, \varepsilon^{(t)}} F(\mathbf{x}^{(t+1)}, \mathbf{h}^{(t+1)}),
\end{equation}
where $\partial_{\mathbf{x}, \varepsilon^{(t)}} F(\mathbf{x}^{(t+1)}, \mathbf{h}^{(t+1)})$ is the corresponding $\varepsilon^{(t)}$-subgradient. Now, we get
\begin{equation}\label{apdx_ineq:e_x_upperbound}
        \begin{aligned}
                &\|\mathbf{e}_{\mathbf{x}}^{(t+1)}\|_{2}\\
         = &\|\nabla \phi_{\rho, t}(\mathbf{x}^{(t+1)}) - \nabla \psi_{\rho, t}(\tilde{\mathbf{x}}^{(t)} | \mathbf{x}^{(t)}) - \beta^{(t)} (\mathbf{x}^{(t+1)} - \tilde{\mathbf{x}}^{(t)})\|_{2} \\
         \leq &\|\nabla \phi_{\rho, t}(\mathbf{x}^{(t+1)}) - \nabla \psi_{\rho, t}(\tilde{\mathbf{x}}^{(t)} | \mathbf{x}^{(t)})\|_{2} + \beta^{(t)} \| \mathbf{x}^{(t+1)} - \tilde{\mathbf{x}}^{(t)}\|_{2}.
        \end{aligned}
\end{equation}
The first term of \eqref{apdx_ineq:e_x_upperbound} can be upper bounded by
\begin{multline}\label{apdx_ineq:1st_e_x_upperbound}
        \|\nabla \phi_{\rho, t}(\mathbf{x}^{(t+1)}) - \nabla \psi_{\rho, t}(\tilde{\mathbf{x}}^{(t)} | \mathbf{x}^{(t)})\|_{2}\\
        \leq \|\nabla \phi_{\rho, t}(\mathbf{x}^{(t+1)}) - \nabla \phi_{\rho, t}(\mathbf{x}^{(t)})\|_{2} \\
        + \|\nabla \psi_{\rho, t}(\mathbf{x}^{(t)} | \mathbf{x}^{(t)}) - \nabla \psi_{\rho, t}(\tilde{\mathbf{x}}^{(t)} | \mathbf{x}^{(t)})\|_{2}\\
        \leq (\sqrt{\beta^{(t)}} + 2 \rho) \|\mathbf{x}^{(t+1)} - \mathbf{x}^{(t)}\|_{2} + \sqrt{\beta^{(t)}} \iota^{(t)} \|\mathbf{x}^{(t)} - \mathbf{x}^{(t-1)}\|_{2}.
\end{multline}
The second term of \eqref{apdx_ineq:e_x_upperbound} can be upper bounded by
\begin{multline}\label{apdx_ineq:2nd_e_x_upperbound}
        \beta^{(t)} \| \mathbf{x}^{(t+1)} - \tilde{\mathbf{x}}^{(t)}\|_{2}\\
        \leq \beta^{(t)} \|\mathbf{x}^{(t+1)} - \mathbf{x}^{(t)}\|_{2} + \beta^{(t)} \iota^{(t)} \|\mathbf{x}^{(t)} - \mathbf{x}^{(t-1)}\|_{2}.
\end{multline}
Combining the results in \eqref{apdx_ineq:1st_e_x_upperbound} and \eqref{apdx_ineq:2nd_e_x_upperbound}, we have
\begin{equation}
        \|\mathbf{e}_{\mathbf{x}}^{(t+1)}\|_{2} \leq C_{x}^{(t)} (\|\mathbf{x}^{(t+1)} - \mathbf{x}^{(t)}\|_{2} + \|\mathbf{x}^{(t)} - \mathbf{x}^{(t-1)}\|_{2}),
\end{equation}
where $C_{x}^{(t)} = 2 \rho + (1 + \iota^{(t)}) (\beta^{(t)} + \sqrt{\beta^{(t)}})$ is finite for all $t$.

Similarly, we can bound $\|\mathbf{e}_{\mathbf{h}}^{(t+1)}\|_{2}$ by
\begin{multline}\label{aped_ineq:e_h_upperbound}
        \|\mathbf{e}_{\mathbf{h}}^{(t+1)}\|_{2} \\
        = \|\nabla \gamma_{t+1}(\mathbf{h}^{(t+1)}) - \alpha^{(t)} (\mathbf{h}^{(t+1)} - \tilde{\mathbf{h}}^{(t)} + \bm{\delta}^{(t)}) - \nabla \gamma_{t}(\tilde{\mathbf{h}}^{(t)})\|_{2} \\
        \leq \|\nabla \gamma_{t+1}(\mathbf{h}^{(t+1)}) - \nabla \gamma_{t}(\tilde{\mathbf{h}}^{(t)})\|_{2} + \alpha^{(t)} \|\mathbf{h}^{(t+1)} - \tilde{\mathbf{h}}^{(t)}\|_{2} \\
        + \sqrt{2 \alpha^{(t)} \varepsilon^{(t)}}.
\end{multline}
The first term of \eqref{aped_ineq:e_h_upperbound} can be upper bounded by
\begin{equation}
        \hspace{-0.75em}\begin{aligned}
              &\|\nabla \gamma_{t+1}(\mathbf{h}^{(t+1)}) - \nabla \gamma_{t}(\tilde{\mathbf{h}}^{(t)})\|_{2}\\
        \leq &(\|(\mathbf{S}^{(t+1)} + \mathbf{S}^{(t)})^{\operatorname{H}} \mathbf{h}^{(t+1)}\|_{2} + \|\mathbf{r}\|_{2}) \|\mathbf{S}^{(t+1)} - \mathbf{S}^{(t)}\|_{2}\\
        + &\sqrt{\alpha^{(t)}} \|\mathbf{h}^{(t+1)} - \mathbf{h}^{(t)}\|_{2} + \mu^{(t)} \sqrt{\alpha^{(t)}} \|\mathbf{h}^{(t)} - \mathbf{h}^{(t-1)}\|_{2}.  
        \end{aligned}
\end{equation}
The second term of \eqref{aped_ineq:e_h_upperbound} can be upper bounded by
\begin{multline}
        \alpha^{(t)} \|\mathbf{h}^{(t+1)} - \tilde{\mathbf{h}}^{(t)}\|_{2}\\
        \leq \alpha^{(t)} \|\mathbf{h}^{(t+1)} - \mathbf{h}^{(t)}\|_{2} + \alpha^{(t)} \mu^{(t)} \|\mathbf{h}^{(t)} - \mathbf{h}^{(t-1)}\|_{2}.
\end{multline}
Since $\|\mathbf{S}^{(t+1)} - \mathbf{S}^{(t)}\|_{2} = \|\mathbf{x}^{(t+1)} - \mathbf{x}^{(t)}\|_{2}$ according to \eqref{sec2_eq:Zak_OTFS_modulation}, we get
\begin{multline}
        \|\mathbf{e}_{\mathbf{h}}^{(t+1)}\|_{2} \leq C_{h}^{(t)} (\| \mathbf{h}^{(t+1)} - \mathbf{h}^{(t)} \|_{2} + \|\mathbf{h}^{(t)} - \mathbf{h}^{(t-1)}\|_{2}) \\
        + C_{s}^{(t)} \|\mathbf{x}^{(t+1)} - \mathbf{x}^{(t)}\|_{2} + \sqrt{2 \alpha^{(t)} \varepsilon^{(t)}},
\end{multline}
where both $C_{h}^{(t)} = (\alpha^{(t)} + \sqrt{\alpha^{(t)}}) (1 + \mu^{(t)})$ and $C_{s}^{(t)} = \|(\mathbf{S}^{(t+1)} + \mathbf{S}^{(t)})^{\operatorname{H}} \mathbf{h}^{(t+1)}\|_{2} + \|\mathbf{r}\|_{2}$ are finite.

Therefore, the distance between $\mathbf{0}$ and the $\varepsilon^{(t)}$-subgradient of $F(\mathbf{x}, \mathbf{h})$ at $(\mathbf{x}^{(t+1)}, \mathbf{h}^{(t+1)})$ is upper bounded by
\begin{multline}\label{apdx_ineq:epsilon_subgradient_upperbound}
        \operatorname{dist}(\mathbf{0}, \partial_{\varepsilon^{(t)}} F(\mathbf{x}^{(t+1)}, \mathbf{h}^{(t+1)})) \\
        \leq \bar{C} (\|\mathbf{x}^{(t+1)} - \mathbf{x}^{(t)}\|_{2} + \|\mathbf{x}^{(t)} - \mathbf{x}^{(t-1)}\|_{2} + \| \mathbf{h}^{(t+1)} - \mathbf{h}^{(t)} \|_{2} \\
        + \|\mathbf{h}^{(t)} - \mathbf{h}^{(t-1)}\|_{2}) + \sqrt{2 \alpha^{(t)} \varepsilon^{(t)}},
\end{multline}
where $\bar{C} = \sup_{t} \{C_{h}^{(t)}, C_{s}^{(t)}, C_{x}^{(t)}\}$ is finite.

\subsection{Upper Boundness of Objective Function Descents}

\subsubsection{Upper Boundness of $f(\mathbf{x}^{(t)}, \mathbf{h}^{(t+1)}) - f(\mathbf{x}^{(t)}, \mathbf{h}^{(t)})$}

Since Algorithm \ref{sec4_alg:AST} solves \eqref{sec4_prob:atomic_mix_int} into $\varepsilon^{(t)}$-optimally, we then have
\begin{equation}\label{apdx_ineq:epsilon_optimal}
        f(\mathbf{x}^{(t)}, \mathbf{h}^{(t+1)}) \leq \varepsilon^{(t)} + \inf_{\mathbf{h}} \left\{ f(\mathbf{x}^{(t)}, \mathbf{h}) \right\}.
\end{equation}
According to the result in \cite{Coordinate_ANM_Li_2024}, Algorithm \ref{sec4_alg:AST} arrives \eqref{apdx_ineq:epsilon_optimal} within finite iterations. Since $\gamma_{t}(\mathbf{h})$ in \eqref{sec4_prob:AST_h} is convex and $\alpha^{(t)}$-Lipschitz, $\gamma_{t}(\mathbf{h}^{(t+1)})$ can be bounded as follows:
\begin{multline}\label{apdx_ineq:upperbound_gamma}
        \gamma_{t}(\mathbf{h}^{(t+1)}) \\
        \leq \gamma_{t}(\tilde{\mathbf{h}}^{(t)}) + \langle \nabla \gamma_{t}(\tilde{\mathbf{h}}^{(t)}), \mathbf{h}^{(t+1)} - \tilde{\mathbf{h}}^{(t)} \rangle + \frac{\alpha^{(t)}}{2} \|\mathbf{h}^{(t+1)} - \tilde{\mathbf{h}}^{(t)}\|_{2}^{2}\\
        \leq \gamma_{t}(\mathbf{h}^{(t)}) + \langle \nabla \gamma_{t}(\tilde{\mathbf{h}}^{(t)}), \tilde{\mathbf{h}}^{(t)} - \mathbf{h}^{(t)} \rangle + \langle \nabla \gamma_{t}(\tilde{\mathbf{h}}^{(t)}), \mathbf{h}^{(t+1)} - \tilde{\mathbf{h}}^{(t)} \rangle\\
         + \frac{\alpha^{(t)}}{2} \|\mathbf{h}^{(t+1)} - \tilde{\mathbf{h}}^{(t)}\|_{2}^{2}.
        % \gamma_{t}(\mathbf{h}^{(t+1)}) \leq \gamma_{t}(\mathbf{h}^{(t)}) + \langle \nabla \gamma_{t}(\tilde{\mathbf{h}}^{(t)}), \mathbf{h}^{(t+1)} - \mathbf{h}^{(t)} \rangle \\
        % + \frac{\alpha^{(t)}}{2} \|\mathbf{h}^{(t+1)} - \tilde{\mathbf{h}}^{(t)}\|_{2}^{2}
\end{multline}
Since $\mathbf{h}^{(t+1)}$ is an $\varepsilon^{(t)}$-optimal solution of \eqref{sec4_prob:AST_h_proximal}, according to Lemma \ref{lm:epsilon_solution_property}, we have
\begin{multline}\label{apdx_ineq:upperbound_varphi}
        \varphi(\mathbf{h}^{(t+1)}) \leq \varphi(\mathbf{h}^{(t)}) - \langle \nabla \gamma_{t}(\tilde{\mathbf{h}}^{(t)})\\
        + \alpha^{(t)} (\bm{\delta}^{(t)} + \mathbf{h}^{(t+1)} - \tilde{\mathbf{h}}^{(t)}),  \mathbf{h}^{(t+1)} - \mathbf{h}^{(t)}\rangle + \varepsilon^{(t)}.
\end{multline}
By adding \eqref{apdx_ineq:upperbound_gamma} and \eqref{apdx_ineq:upperbound_varphi} together and substituting $\tilde{\mathbf{h}}^{(t)}$ in \eqref{sec4_eq:mid_h_acc_proximal} into the result, after some calculations, it gives:
\begin{multline}\label{apdx_ineq:upperbound_h_update}
        f(\mathbf{x}^{(t)}, \mathbf{h}^{(t+1)}) - f(\mathbf{x}^{(t)}, \mathbf{h}^{(t)}) \leq \sqrt{2 \alpha^{(t)} \varepsilon^{(t)}} \|\mathbf{h}^{(t+1)} - \mathbf{h}^{(t)}\|_{2} \\
        - \frac{\alpha^{(t)}}{2} \left(\|\mathbf{h}^{(t+1)} - \mathbf{h}^{(t)}\|_{2}^{2} - \bar{\mu}^2\|\mathbf{h}^{(t)} - \mathbf{h}^{(t-1)}\|_{2}^{2}\right) + \varepsilon^{(t)},
\end{multline}
where $\|\bm{\delta}^{(t)}\|_{2} \leq \sqrt{\frac{2 \varepsilon^{(t)}}{\alpha^{(t)}}}$ is implemented with Cauchy-Schwartz inequality to further upperbound the result.

\subsubsection{Upper Boundness of $f(\mathbf{x}^{(t+1)}, \mathbf{h}^{(t+1)}) - f(\mathbf{x}^{(t)}, \mathbf{h}^{(t+1)})$}
The following lemma is useful to derive the upper boundness.
\begin{lemma}[\cite{xu2013block}]\label{apdx_lm:3}
        Let
        \begin{equation}
                \mathbf{x}^{+} = \Pi_{\mathcal{X}}\left(\mathbf{z} - \frac{1}{\beta} \nabla H(\mathbf{z})\right),
        \end{equation}
        where $\mathbf{z} = \mathbf{x} + \alpha (\mathbf{x} - \bar{\mathbf{x}})$ with $\mathbf{x}$, $\bar{\mathbf{x}} \in \mathcal{X}$, $\alpha \geq 0$; $H$ is convex and has Lipschitz continuous gradient; $\mathcal{X}$ is convex; $\beta$ is chosen to satisfy
        \begin{equation}
                H(\mathbf{x}^{+}) \leq H(\mathbf{z}) + \langle \nabla H(\mathbf{z}), \mathbf{x}^{+} - \mathbf{z} \rangle + \frac{\beta}{2} \|\mathbf{x}^{+} - \mathbf{z}\|_{2}^{2}.
        \end{equation}
        Then, it holds that
        \begin{equation}
                H(\mathbf{x}) - H(\mathbf{x}^{+}) \geq \frac{\beta}{2} (\|\mathbf{x}^{+} - \mathbf{x}\|_{2}^{2} - \alpha^2 \|\mathbf{x} - \bar{\mathbf{x}}\|_{2}^{2}).
        \end{equation}
\end{lemma}
According to the update rule \eqref{sec4_eq:update_x} and Lemma \ref{apdx_lm:3}, we have
\begin{equation}\label{apdx_ineq:upperbound_x_update}
        \begin{aligned}
                     & f(\mathbf{x}^{(t+1)}, \mathbf{h}^{(t+1)}) - f(\mathbf{x}^{(t)}, \mathbf{h}^{(t+1)})                                                            \\
                \leq & \psi_{\rho, t}(\mathbf{x}^{(t+1)} | \mathbf{x}^{(t)}) - \psi_{\rho, t}(\mathbf{x}^{(t)} | \mathbf{x}^{(t)})                                    \\
                \leq & - \frac{\beta^{(t)}}{2} (\|\mathbf{x}^{(t+1)} - \mathbf{x}^{(t)}\|_{2}^{2} - \bar{\iota}^2 \|\mathbf{x}^{(t)} - \mathbf{x}^{(t-1)}\|_{2}^{2}).
        \end{aligned}
\end{equation}

\subsection{Convergence Speed Towards $\varepsilon$-Stationary Point}

Adding the results in \eqref{apdx_ineq:upperbound_h_update} and \eqref{apdx_ineq:upperbound_x_update} from $0$ to $t$, we get
\begin{multline}\label{apdx_ineq:lower_bound_summation_f_0_t}
        f(\mathbf{x}^{(0)}, \mathbf{h}^{(0)}) - f^{\star} \geq f(\mathbf{x}^{(0)}, \mathbf{h}^{(0)}) - f(\mathbf{x}^{(t+1)}, \mathbf{h}^{(t+1)})\\
        \geq \sum_{t^{\prime}=0}^{t} \frac{\alpha^{(t^{\prime})}}{2} \left(\|\mathbf{h}^{(t^{\prime}+1)} - \mathbf{h}^{(t^{\prime})}\|_{2}^{2} - \bar{\mu}^2\|\mathbf{h}^{(t^{\prime})} - \mathbf{h}^{(t^{\prime}-1)}\|_{2}^{2} \right) \\
        - \sqrt{2 \alpha^{(t^{\prime})} \varepsilon^{(t^{\prime})}} \|\mathbf{h}^{(t^{\prime}+1)} - \mathbf{h}^{(t^{\prime})}\|_{2}  - \varepsilon^{(t^{\prime})} \\
        + \frac{\beta^{(t^{\prime})}}{2} \left(\|\mathbf{x}^{(t+1)} - \mathbf{x}^{(t^{\prime})}\|_{2}^{2} - \bar{\iota}^2 \|\mathbf{x}^{(t^{\prime})} - \mathbf{x}^{(t^{\prime}-1)}\|_{2}^{2}\right),
\end{multline}
where $f^{\star} > -\infty$ is the lower bound of the objective function. For the first term in \eqref{apdx_ineq:lower_bound_summation_f_0_t}, we have the result in \eqref{apdx_ineq:lower_bound_summation_f_0_t_1}.
\begin{figure*}
        \small\begin{multline}\label{apdx_ineq:lower_bound_summation_f_0_t_1}
                \sum_{t^{\prime}=0}^{t} \frac{\alpha^{(t^{\prime})}}{2} \left(\|\mathbf{h}^{(t^{\prime}+1)} - \mathbf{h}^{(t^{\prime})}\|_{2}^{2} - \bar{\mu}^2\|\mathbf{h}^{(t^{\prime})} - \mathbf{h}^{(t^{\prime}-1)}\|_{2}^{2} \right) = \sum_{t^{\prime}=0}^{t-1} \frac{\alpha^{(t^{\prime})} - \bar{\mu}^2 \alpha^{(t^{\prime}+1)}}{2} \|\mathbf{h}^{(t^{\prime}+1)} - \mathbf{h}^{(t^{\prime})}\|_{2}^{2} + \frac{\alpha^{(t)}}{2} \|\mathbf{h}^{(t+1)} - \mathbf{h}^{(t)}\|_{2}^{2}\\
                \geq \sum_{t^{\prime}=0}^{t} \frac{\alpha^{(t^{\prime})} - \bar{\mu}^2 \alpha^{(t^{\prime}+1)}}{2} \|\mathbf{h}^{(t^{\prime}+1)} - \mathbf{h}^{(t^{\prime})}\|_{2}^{2} \geq \frac{\munderbar{C}_{h} \theta}{2} \sum_{t^{\prime}=0}^{t} \|\mathbf{h}^{(t^{\prime}+1)} - \mathbf{h}^{(t^{\prime})}\|_{2}^{2}.
        \end{multline}
\end{figure*}
Similarly, we also have the result
\begin{multline}\label{apdx_ineq:lower_bound_summation_f_0_t_2}
        \sum_{t^{\prime}=0}^{t} \frac{\beta^{(t^{\prime})}}{2} \left(\|\mathbf{x}^{(t+1)} - \mathbf{x}^{(t^{\prime})}\|_{2}^{2} - \bar{\iota}^2 \|\mathbf{x}^{(t^{\prime})} - \mathbf{x}^{(t^{\prime}-1)}\|_{2}^{2}\right) \\
        \geq \frac{\munderbar{C}_{s}\theta}{2} \sum_{t^{\prime}=0}^{t} \|\mathbf{x}^{(t+1)} - \mathbf{x}^{(t^{\prime})}\|_{2}^{2}.
\end{multline}
By substituting the results in \eqref{apdx_ineq:lower_bound_summation_f_0_t_1} and \eqref{apdx_ineq:lower_bound_summation_f_0_t_2} into \eqref{apdx_ineq:lower_bound_summation_f_0_t}, after some calculations, we get the result in \eqref{apdx_ineq:summation_result}.
\begin{figure*}
        \small\begin{multline}\label{apdx_ineq:summation_result}
                f(\mathbf{x}^{(0)}, \mathbf{h}^{(0)}) - f(\mathbf{x}^{(t+1)}, \mathbf{h}^{(t+1)}) \geq  \frac{\munderbar{C}_{h} \theta}{2} \sum_{t^{\prime}=0}^{t} \|\mathbf{h}^{(t^{\prime}+1)} - \mathbf{h}^{(t^{\prime})}\|_{2}^{2}  - \sqrt{2 \alpha^{(t^{\prime})} \varepsilon^{(t^{\prime})}} \|\mathbf{h}^{(t^{\prime}+1)} - \mathbf{h}^{(t^{\prime})}\|_{2} + \frac{\munderbar{C}_{s}\theta}{2} \sum_{t^{\prime}=0}^{t} \|\mathbf{x}^{(t+1)} - \mathbf{x}^{(t^{\prime})}\|_{2}^{2}  - \varepsilon^{(t^{\prime})}\\
                = \frac{\munderbar{C}_{h} \theta}{2} \sum_{t^{\prime} = 0}^{t} \underbrace{\left(\|\mathbf{h}^{(t^{\prime}+1)} - \mathbf{h}^{(t^{\prime})}\|_{2} - \sqrt{2 \alpha^{(t^{\prime})} \varepsilon^{(t^{\prime})} }/ (\munderbar{C}_{h} \theta) \right)^2}_{\text{denote as } h_{t^{\prime}+1}^2} + \frac{\munderbar{C}_{s}\theta}{2} \sum_{t^{\prime}=0}^{t} \|\mathbf{x}^{(t^{\prime}+1)} - \mathbf{x}^{(t^{\prime})}\|_{2}^{2} - \underbrace{\sum_{t^{\prime} = 0}^{t} \left(1 + \frac{1}{\munderbar{C}_{h} \theta}\right) \varepsilon^{(t^{\prime})}}_{\text{denote as } C_{\varepsilon}}.
        \end{multline}
\end{figure*}
From \eqref{apdx_ineq:summation_result}, we get
\begin{multline}
        f(\mathbf{x}^{(0)}, \mathbf{h}^{(0)}) - f^{\star} + C_{\varepsilon} \geq \min_{t^{\prime} = 0, 1, \dots, t} \frac{\munderbar{C}_{h} \theta t}{4} (h_{t^{\prime}+1}^2+ h_{t^{\prime}}^2) \\
        + \frac{\munderbar{C}_{s} \theta t}{4} (\|\mathbf{x}^{(t^{\prime}+1)} - \mathbf{x}^{(t^{\prime})}\|_{2}^{2} + \|\mathbf{x}^{(t^{\prime})} - \mathbf{x}^{(t^{\prime}-1)}\|_{2}^{2}).
\end{multline}
By using Cauchy-Schwartz inequality, we have
\begin{multline}\label{apdx_ineq:sqrt_func_val_bound}
        \sqrt{f(\mathbf{x}^{(0)}, \mathbf{h}^{(0)}) - f^{\star} + C_{\varepsilon}}\\
        \geq \min_{t^{\prime} = 0, 1, \dots, t} \frac{\sqrt{\munderbar{C}_{h, s} \theta t}}{4} \left(h_{t^{\prime}+1} + h_{t^{\prime}} + \|\mathbf{x}^{(t^{\prime}+1)} - \mathbf{x}^{(t^{\prime})}\|_{2} \right.\\
        \left.+ \|\mathbf{x}^{(t^{\prime})} - \mathbf{x}^{(t^{\prime}-1)}\|_{2}\right),
\end{multline}
where $\munderbar{C}_{h, s} = \min\{\munderbar{C}_{h}, \munderbar{C}_{s}\}$. For the term $\sqrt{t} h_{t^{\prime}+1}$, since we require the sequence $\{\varepsilon^{(t)}\}_{t \geq 0}$ to be summable and decreasing, there exists a constant number $\bar{E}$ to be such that
\begin{equation}\label{apdx_ineq:t_epsilon_bound}
        \sqrt{t} \min_{t^{\prime} = 0, 1, \dots, t} \sqrt{2 \alpha^{(t^{\prime})} \varepsilon^{(t^{\prime})}} \leq \sqrt{2 t \bar{\alpha} \varepsilon^{(t)}}  \leq \bar{E} < \infty, \forall t.
\end{equation}
Therefore, combine the results in \eqref{apdx_ineq:epsilon_subgradient_upperbound}, \eqref{apdx_ineq:sqrt_func_val_bound}, and \eqref{apdx_ineq:t_epsilon_bound} together, we finally get
\begin{equation}
        \operatorname{dist}(\mathbf{0}, \partial_{\varepsilon^{(t)}} F(\mathbf{x}^{(t+1)}, \mathbf{h}^{(t+1)})) \leq \frac{C}{\sqrt{t}},
\end{equation}
where
\begin{equation}
        C = 4 \bar{C}\sqrt{\frac{f(\mathbf{x}^{(0)}, \mathbf{h}^{(0)}) - f^{\star} + C_{\varepsilon}}{\munderbar{C}_{h, s} \theta}} + \left(1 + \frac{2}{\munderbar{C}_{h, s} \theta}\right)\bar{E}.
\end{equation}

}

\bibliographystyle{IEEEtran}
\bibliography{refs}
\vfill
\end{document}